\documentclass[journal]{IEEEtran}
\usepackage{lineno}
\usepackage{cite}
\usepackage{hyperref}
\usepackage[tbtags]{amsmath}
\usepackage{amssymb,amsfonts}
\usepackage{amsthm}
\usepackage{amsbsy}

\usepackage{graphicx}
\usepackage{textcomp}
\usepackage{xcolor}
\usepackage{float}
\usepackage{mathrsfs}
\usepackage{mathtools}
\usepackage{multirow}
\usepackage{algpseudocode}
\usepackage{comment}
\usepackage[inline, shortlabels]{enumitem}
\makeatletter
\newif\if@restonecol

\usepackage{algorithm}

\usepackage{caption}
\usepackage{subcaption}
\usepackage{ragged2e}

\usepackage{dsfont}
\usepackage{bbm}
\usepackage{stfloats}
\usepackage{bm}

\newtheorem{corollary}{Corollary}

\theoremstyle{definition}
\newtheorem{theorem}{Theorem}

\newtheorem{lemma}{Lemma}

\newtheorem{proposition}{Proposition}

\newcommand{\biggg}{\bBigg@{3}}
\newcommand{\Biggg}{\bBigg@{3.5}}
\makeatother
\begin{document}

\title{ Rotatable Array-Aided Hybrid Beamforming for Integrated Sensing and Communication}

\author{Zequan~Wang, Liang~Yin, Zimeng~Lei, Yitong~Liu, Yunan~Sun, and Hongwen~Yang

\thanks{Z. Wang, L. Yin, Z. Lei, Y. Liu, Y. Sun, and H. Yang are with the School of Information and Communication Engineering, Beijing University of Posts and Telecommunications, Beijing, 100876, China (e-mail: \{zequanwang, YinL, leizimeng, liuyitong, sunyunan, yanghong\}@bupt.edu.cn).}
\thanks{(Corresponding author: Liang~Yin)}}

\maketitle

\begin{abstract}
Six-dimensional movable antenna (6DMA) technology has been proposed to enhance the performance of Integrated Sensing and Communication (ISAC) systems. 
However, within 6DMA-related research, studies on the ISAC system based on rotatable array (RA) remains relatively limited. 
Given the significant advantages of hybrid beamforming technology in balancing system performance and hardware complexity, this paper focuses on a channel model that accounts for the efficiency of the antenna radiation pattern and studies the sub-connected hybrid beamforming design for multi-user RA-aided ISAC systems. 
Aiming at the non-convex nature with coupled variables in this problem, this paper transforms the complex fractional objective function using the Fractional Programming (FP) method, and then proposes an algorithm based on the Alternating Optimization (AO) framework, which achieves optimization by alternately solving five subproblems. 
For the analog beamforming optimization subproblem, we adopt Singular Value Decomposition (SVD) method to transform the objective function, thereby deriving the closed-form update expression for the analog beamforming matrix. For the antenna rotation optimization subproblem, we derive the closed-form derivative expression of the array rotation angle and propose a two-stage Gradient Ascent (GA) based method to optimize the antenna rotation angle. 
Extensive simulation results demonstrate the effectiveness of the proposed RA-aided hybrid beamforming design method. It not only significantly improves the overall system performance while reducing hardware costs, but also achieves performance comparable to that of the fully-digital beamforming design with fixed-position antennas (FPA) under specific parameter configurations. 
\end{abstract}


\section{Introduction}
Integrated sensing and communication (ISAC) has become a highly promising technology in the sixth-generation (6G) wireless network \cite{Liu2018}\cite{Liu2022}. Its core advantage is the ability to share hardware platforms and signal bandwidths, which helps achieve efficient resource use for communication and radar sensing. This feature not only greatly improves spectrum efficiency, but also cuts hardware costs substantially \cite{Liu2020}. Due to these advantages, ISAC has shown great potential in many applications. These include vehicle-to-everything (V2X), industrial internet of things (IIoT), and environmental monitoring \cite{Liu2022}. In addition, ISAC has attracted wide attention from both academia and industry. It integrates a series of emerging wireless technologies, such as reconfigurable intelligent surface (RIS) \cite{Sun2025} \cite{He2022}, unmanned aerial vehicle (UAV) \cite{Meng2024}, and non-orthogonal multiple access (NOMA) \cite{Wang2022}\cite{Lyu2024}. These technologies further enhance ISAC's performance and application potential.

With the continuous evolution of 6G communication systems, the large-scale application of high-frequency millimeter wave and large antenna array has become an irreversible technological trend. By increasing antenna numbers, large antenna array can significantly improve beamforming gain and enhance spatial resolution, thereby effectively improving the communication and sensing performance. However, the fully-digital beamforming scheme requires each antenna unit to be equipped with an independent radio frequency chain, which makes the deployment of large antenna array face the challenge of a sharp rise in hardware costs and power consumption, seriously restricting its popularization and application in practical scenarios. In contrast, the hybrid beamforming \cite{Yu2016}\cite{S2019}, relying on a small number of RF chains and an analog front-end consisting of phase shifters (PSs), has been considered as a promising technique to achieve a good trade-off between cost and communication performance.

Based on the front-end design of PS circuits, the hybrid beamformer can be categorized into fully-connected structures and sub-connected structures \cite{Du2018}. In the fully-connected structure, all RF chains are connected to all antenna elements through PSs, which endows a higher system degree of design freedom (DoFs). However, it also results in high hardware costs and energy consumption overheads, which is particularly prominent in the scenario of large-scale antenna array in millimeter wave systems. In contrast, the sub-connected structure, by connecting each RF chain only to a part of the antenna elements, has higher energy efficiency and is easier for engineering implementation. To improve system performance under the constraint of a limited number of RF chains, numerical optimization algorithms have been proposed\cite{Wang2022}\cite{Yu2016}\cite{Zhu2024}. \cite{Wang2022} proposes a method to directly optimize the hybrid beamformer in an iterative manner. \cite{Yu2016} designs a fully-digital beamformer, then iteratively optimizes the hybrid beamformer to gradually approximate the performance of the fully-digital solution. \cite{Zhu2024} introduces a continuous auxiliary variable to replace the beamforming variable with a unit modulus constraint, optimizes the auxiliary variable, extracts its phase, and finally obtains the design result of the beamformer.
Nevertheless, most of the existing fully-connected and sub-connected hybrid beamforming schemes rely on fixed position antennas (FPA), and the transmitter/receiver fails to fully exploit the spatial variation of wireless channels, thus ignoring the gain improvement that can be brought about by dynamic channel adjustment.

ISAC aims primarily to enhance communication capacity and sensing ability \cite{An2023}. In recent years, movable antenna (MA) and six-dimensional MA (6DMA) technologies have been proposed to improve the performance of wireless communication \cite{ZhuL20241}-\cite{Shao20253}. Specifically, compared with the traditional MA which can only dynamically adjust the position of antenna elements within a specific feasible region, the 6DMA can flexibly regulate its spatial position and rotation angle, and integrate the DoFs of the antenna in three-dimensional (3D) position and rotation dimensions, thus providing customized sensing and communication services more efficiently \cite{Shao2025}-\cite{Shao20253}. Several studies have investigated the integration of MA to enhance communication capacity \cite{ZhuL20242}-\cite{Sun20252}. Specifically, the authors analyze the performance of the MA-enabled wireless communication system in \cite{ZhuL20242}. They derive closed-form expressions for channel gain under deterministic and random channel environments, verifying the effectiveness of MA in improving channel capacity.  
\cite {Ding2025} investigates an MA-enabled full-duplex ISAC system, which considers the scenario of user uplink and downlink transmissions and sets the sum of user uplink/downlink communication rate and target sensing rate as the optimization objective.
The authors in \cite{Lyu2025} conduct research on ISAC systems with clutter interference, deriving in detail the derivative of channel gain with respect to antenna position, and optimizing the antenna position based on the gradient ascent method.
The authors in \cite{Zhang2025}\cite{Zhang2024} propose a scheme for MA-aided multi-user communication scenarios under hybrid precoding design.
The authors in \cite{Sun20251} study the trade-off between communication and sensing in MA-enabled ISAC system, and adopts the Particle Swarm Optimization (PSO) algorithm to determine the optimal position of the antenna.
\cite{Zar2024} explores the application of MA in multi-target detection, and the results show that MA can significantly reduce transmission power while ensuring communication and sensing performance.
\cite{Sun20252} investigates multi-user communication systems supported by the collaboration of RIS and MA.

Although the advantages of MA in ISAC systems have been validated, effectively enhancing communication and sensing performance by optimizing antenna positions remains a challenging issue. The main difficulty arises from the complex expression of the wireless channel with respect to the positions of the antennas, causing non-convex expressions involving coupling variables.
Existing studies mostly focus on MA-enabled ISAC systems with adjustable positions. In contrast, the investigation of ISAC systems considering antenna rotatability is still in its infancy. Generally, antenna rotation can optimize beam pointing more flexibly, and can provide customized sensing and communication services more efficiently, especially when considering a channel model that accounts for the efficiency of the antenna radiation pattern \cite{Wang2023}.
Thus, this paper studies the hybrid beamforming design for multi-user multi-input-single-output (MU-MISO) ISAC systems based on rotatable arrays, aiming to reduce hardware cost and complexity while improving the communication and sensing performance of the system.
To the best of the authors' knowledge, this is among the first works to explore RA-based ISAC systems. The main contributions of this paper are summarized as follows.  
\begin{enumerate}
    \item We study a hybrid beamforming design for an RA-aided ISAC system, where the channel model accounts for the efficiency of the antenna radiation pattern. To fully exploit the performance gains brought by the array rotation to the ISAC system, flexibly rotatable transmit and receive antenna panels are deployed at the BS to provide communication services for multiple users while sensing a target. 
    We aim to maximize the weighted sum of each user's communications rate and the target sensing rate by jointly optimizing the transmit digital beamformer, transmit analog beamformer, receive beamformer, and antenna rotation angles to balance sensing accuracy and communication efficiency.
    \item We propose an alternating optimization (AO) method to solve the formulated non-convex optimization problem with highly-coupled variables. Specifically, the minimum mean square error (MMSE) filtering method is first adopted to optimize the receive beamformer. Then, facilitated by fractional programming (FP), the remaining problem is decomposed into four subproblems to handle the coupled variables. 
    For the hybrid beamforming optimization, we adopt singular value decomposition (SVD) and Karush-Kuhn-Tucker (KKT) conditions to derive the closed-form solutions, thereby solving the analog beamforming optimization problem and the digital beamforming optimization problem. 
    Then the closed-form derivative expression of the array rotation angle is derived, and a gradient ascent (GA) method is proposed to solve the array rotation angle optimization problem.
    \item We perform numerical simulations to evaluate the advantages of sub-connected RA-aided ISAC system and the effectiveness of the proposed algorithm. Simulation results demonstrate that the RA-aided ISAC system outperforms the FPA-based ISAC system due to the additional DoFs introduced by array rotation optimization.
    Moreover, under certain practical conditions, the sub-connected RA-aided ISAC system yields higher communication and sensing performance compared to the fully-connected FPA-based system and the fully-digital FPA-based system. 
\end{enumerate}

Notation: $a/A$, $\mathbf{a}$, $\mathbf{A}$, and $\mathcal{A}$ denote a scalar, a vector, a matrix, and a set, respectively. $(\cdot)^T$, $(\cdot)^H$, $\|\cdot\|_2$, $|\cdot|$, $\|\cdot\|_F$, $\operatorname{Tr}\{\cdot\}$, and $\operatorname{Rank}\{\cdot\}$ denote the transpose, conjugate transpose, Euclidean norm, absolute value, Frobenius matrix norm, trace, and rank, respectively. $j=\sqrt{-1}$ represents the imaginary unit. $\mathbb{C}^{M \times N}$ and $\mathbb{R}^{M \times N}$ are the sets for complex and real matrices of $M \times N$ dimensions, respectively. $\mathbf{I}_N$ is the identity matrix of order $N$. $\mathcal{CN}(0, \sigma^2)$ represents the circularly symmetric complex Gaussian (CSCG) distribution with mean zero and covariance $\sigma^2$. Finally, $[\cdot]_{(m,n)}$ denotes the $(m,n)$-th element of a matrix.
\section{System Model}
\subsection{RA-BS Model}
% \subsection{System Model}\label{s2}
{As illustrated in Fig. \ref{scene}, this paper considers an RA-aided far-field MU-MISO system. The system comprises a dual-functional radar and communication (DFRC) BS serving $K$ single-FPA users, sensing one target, and considering $C$ clutters as interference for sensing. 
The BS is equipped with $N_t$ transmit antennas and $N_r$ receive antennas, which are respectively fixed on a rotatable transmit antenna panel (TP) and a rotatable receive antenna panel (RP), each with an area of $D$. 
To reduce power consumption and hardware complexity of the BS, the BS employs the sub-connected hybrid array, which is divided into $N_{RF}$ nonoverlapped subarrays, each containing $M=N_\mathrm t/N_{RF} $ antennas connected to an RF chain through PSs. $M$ is an integer for brevity.
A bistatic setup is adopted at the BS, where TP and RP are separated to avoid challenging issues in full-duplex mode. Thus, self-interference (SI) can be significantly suppressed when the TP and RP are sufficiently separated. Therefore, this paper assumes that the SI is completely cancelled. 
\begin{figure}[!t]
    \centering
    \includegraphics[height=0.36\textwidth]{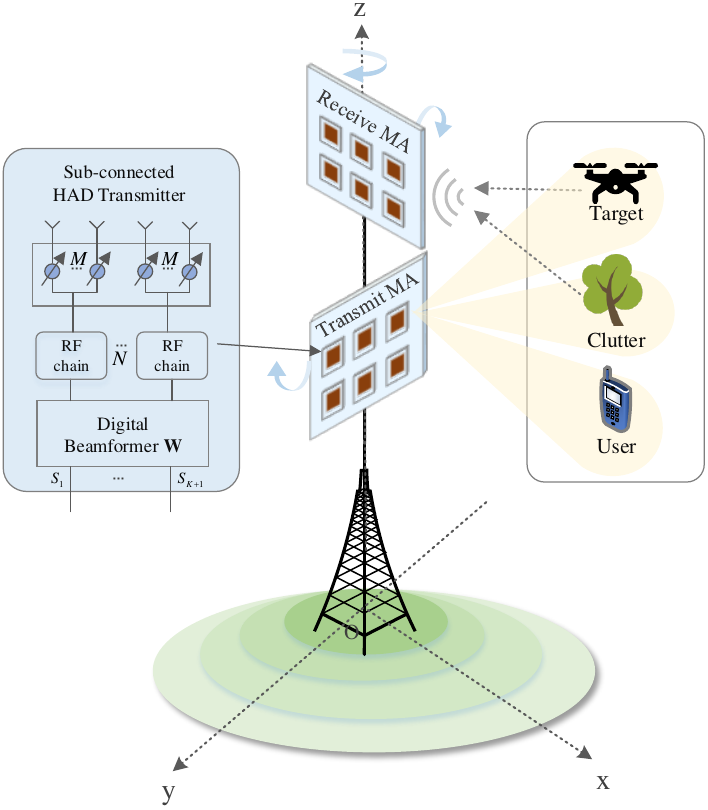}
    \captionsetup{font={small}}
    \caption{\justifying Illustration of the proposed RA-aided far-field ISAC system.}
    \label{scene}
\end{figure}

To avoid coupling effects, the distance between any two antennas is always no less than half a wavelength. Without loss of generality, the antennas lie within the $y-z$ plane of the local Cartesian coordinate system defined by the center of the corresponding plane, and the local positions of the $N_t$-th transmit antenna and the $N_r$-th receive antenna can be respectively represented as
\begin{equation}
\begin{aligned}
\mathbf q_{n_t}^\mathrm t=[0,y^\mathrm t_{n_t},z^\mathrm t_{n_t}]^\mathrm T \in A_t,
\end{aligned}
\end{equation}
\begin{equation}
\begin{aligned}
\mathbf q_{n_r}^\mathrm r=[0,y^\mathrm r_{n_r},z^\mathrm r_{n_r}]^\mathrm T \in A_r.
\end{aligned}
\end{equation}

The transmit and receive arrays are connected to the CPU via rotatable rods embedded with flexible wires, and thus 3D rotations can be adjusted. The 3D rotation can be respectively characterized by
\begin{equation}
\begin{aligned}
\boldsymbol\varpi^\mathrm q=[\alpha^\mathrm q,\beta^\mathrm q,\gamma^\mathrm q]^\mathrm T, q \in\{\mathrm t, \mathrm r\},
\end{aligned}
\end{equation}
where $\alpha^\mathrm q \in [0,2\pi)$, $\beta^\mathrm q \in [0,2\pi)$, $\gamma^\mathrm q \in [0,2\pi)$ denote the rotation angles with respect to (w.r.t.) the $x$-axis, $y$-axis and $z$-axis, respectively.

Given $\boldsymbol\varpi$, the rotation matrix can be defined as

\begin{equation}
\begin{split}
&\mathbf{R}(\boldsymbol{\varpi}^q) = \mathbf R_{\alpha_q} \mathbf R_{\beta_q} \mathbf R_{\gamma_q}\\
&=\begin{bmatrix}
c_{\beta_q}c_{\gamma_q} & c_{\beta_q}s_{\gamma_q} & -s_{\beta_q} \\
s_{\beta_q}s_{\alpha_q}c_{\gamma_q} - c_{\alpha_q}s_{\gamma_q} & s_{\beta_q}s_{\alpha_q}s_{\gamma_q} + c_{\alpha_q}c_{\gamma_q} & c_{\beta_q}s_{\alpha_q} \\
c_{\alpha_q}s_{\beta_q}c_{\gamma_q} + s_{\alpha_q}s_{\gamma_q} & c_{\alpha_q}s_{\beta_q}s_{\gamma_q} - s_{\alpha_q}c_{\gamma_q} & c_{\alpha_q}c_{\beta_q}
\end{bmatrix},
\end{split}
\end{equation}
with
\begin{equation}
\begin{aligned}
    \mathbf R_{\alpha_q}=
\begin{bmatrix}
1 & 0 & 0 \\
0 & c_{\alpha_q} & s_{\alpha_q} \\
0 & - s_{\alpha_q} & c_{\alpha_q}
\end{bmatrix},\\
    \mathbf R_{\beta_q}=
\begin{bmatrix}
c_{\beta_q} & 0 & -s_{\beta_q} \\
0 & 1 & 0 \\
s_{\beta_q} & 0 & c_{\beta_q}

\end{bmatrix},\\
    \mathbf R_{\gamma_q}=
\begin{bmatrix}
c_{\gamma_q} & s_{\gamma_q} & 0 \\
-s_{\gamma_q} & c_{\gamma_q} & 0 \\
0 & 0 & 1
\end{bmatrix},
\end{aligned}
\end{equation}
denoting the rotation matrices w.r.t. each of the $x$-axis, $y$-axis and $z$-axis, respectively, where $c_x$ = cos($x$) and $s_x$ = sin($x$) are defined for notational simplicity.
Then the positions of the $n_t$-th transmit antenna and the $n_r$-th receive antenna in the global Cartesian coordinate system can be expressed as
\begin{equation}
\begin{aligned}
    \mathbf t_{n_\mathrm t}=\mathbf R(\boldsymbol{\varpi}^t)\mathbf q_{n_t}^\mathrm t+\mathbf C^\mathrm t,
\end{aligned}
\end{equation}

\begin{equation}
\begin{aligned}
    \mathbf r_{n_\mathrm r}=\mathbf R(\boldsymbol{\varpi}^r)\mathbf q_{n_r}^\mathrm r+\mathbf C^\mathrm r,
\end{aligned}
\end{equation}
where $\mathbf C^\mathrm t$ and $\mathbf C^\mathrm r$ denote the centers'positions of TP and RP, respectively. 
}
\subsection{Channel Model}
Considering a far-field quasi-static channel, the communication users (CUs) and the target are located in the far-field region of the BS. Since the size of  each UPA's region is much smaller than the signal propagation distance from the BS to the scatter points/receivers, the planar wave is adopted to construct the field-response array model. The pointing vector of the CUs/target corresponding to the direction $(\theta,\phi)$ is given by 
\begin{equation}
\mathbf p=[\sin \theta \cos \phi,\sin \theta \sin \phi,\cos \theta ]^T, 
\end{equation}
where $\theta \in [0, \pi]$ and $\phi \in [-\frac{\pi}{2},\frac{\pi}{2}]$ denote the elevation and azimuth angles, respectively. The difference in signal path distances from the target position $\mathbf p$ to the  $n_t$-th transmit antenna and from $\mathbf p$ to the $n_r$-th receive antenna are represented by $\mathbf p^T \mathbf t_{n_\mathrm t}(\boldsymbol {\varpi}^t) $ and $\mathbf p^T \mathbf r_{n_\mathrm r}(\boldsymbol {\varpi}^r) $, and the corresponding complex forms of phase difference are given by $e^{-j\frac{2\pi}{\lambda}\mathbf p^T \mathbf t_{n_\mathrm t}(\boldsymbol {\varpi}^t) }$ and $e^{-j\frac{2\pi}{\lambda}\mathbf p^T \mathbf r_{n_\mathrm r}(\boldsymbol {\varpi}^r) }$. Taking $\mathbf p$ as the reference point, the steer vector of the TP is
\begin{equation}
\begin{aligned}
    \mathbf a(\boldsymbol {\varpi}^t )=[e^{-j\frac{2\pi}{\lambda}\mathbf p^T \mathbf t_{1}(\boldsymbol {\varpi}^t) },\dots,e^{-j\frac{2\pi}{\lambda}\mathbf p^T \mathbf t_{n_\mathrm t}(\boldsymbol {\varpi}^t) }]^T,
\end{aligned}
\end{equation}
where $\lambda$ denotes the carrier wavelength.

As mentioned in \cite{Shao20253}, the rotatability of both transmit and receive arrays must meet certain rotation constraints to avoid mutual signal reflections between them. In this paper, however, a relatively large distance is maintained between the TP and the RP to address this issue. Therefore, it is assumed that there is no mutual signal reflection.

To make the channel model more general, as well as evaluating the influence of array orientation tuning precisely, we take into account the effect of the antenna radiation pattern \cite{YangS2024}. According to \cite{Wang2023}, since the signals are transmitted by different array elements and received by receiving array elements from distinct angles, there is a significant difference in the pattern of the antennas. The effective antenna pattern is determined by the projection of the array normal to the direction of the signal. Consequently, the channel power gain between the $n_t$-th transmit antenna and the $k$-th CU/target $\mathbf p$ is given by
\begin{equation}
\begin{aligned}
    |h^\mathrm t_{n_t,k}(\mathbf p)|=\rho\sqrt{ \frac{(\mathbf p - \mathbf t _{n_\mathrm t})^\mathrm T \mathbf u^\mathrm t}{||\mathbf p - \mathbf t _{n_\mathrm t}||}},
\end{aligned}
\end{equation}
where $\rho$ is the pass loss for the $k$-th user, $\mathbf u^\mathrm p = \mathbf R(\boldsymbol {\varpi}^\mathrm p) \mathbf u_0$ denotes the outward normal vector of the TP or RP, $\mathbf u_0=[1,0,0]^T$ denotes the original normal vector, and $\frac{(\mathbf p- \mathbf t _{n_\mathrm t})}{||\mathbf p- \mathbf t _{n_\mathrm t}||}$ represents the projection of the array normal to the direction of the signal.

Consequently, The communication channel vector $\mathbf h_k \in \mathbb C^{N_t \times 1}$ between the $k$-th CU and the TP can be derived as
\begin{equation}
    \mathbf h_k= [|h^\mathrm t_{n_1,k}(\mathbf p)|\mathbf [a(\boldsymbol {\varpi}^t )]_1,\dots,|h^\mathrm t_{n_t,k}(\mathbf p)|\mathbf [a(\boldsymbol {\varpi}^t )]_{n_\mathrm t}]^\mathrm T
\end{equation}
and also, the field response vector from the BS to the sensing target/clutters can be written as $\mathbf G_\chi  \in \mathbb C^{N_r \times N_t} $, which is modeled as
\begin{equation}
    \mathbf G_\chi= \beta_\chi \mathbf g_\chi \mathbf h_\chi^H,
\end{equation}
where $\beta_\chi$ represents the target radar cross section (RCS) with $\mathbb E\{|\beta_\chi|^2\}=\zeta^2$, $\chi=\{s,1,\dots,C\}$ represents the target or the clutters, $\mathbf g_\chi$ and $\mathbf h_\chi$ represent the channel vector between the transmitter and target as well as between the receiver and target, respectively.
\subsection{Signal Model}
Let $\mathbf s \in \mathbb C^{(K+1) \times 1}$ be the independent and identically distributed (i.i.d.) signal vector, with $\mathbb E(\mathbf {ss}^\mathrm H) = \mathbf I_{\mathrm K+1}$, where $s_1,\dots,s_K$ are for $K$ CUs, respectively, and $s_{K+1}$ is dedicated for sensing. Assume that the perfect channel state information (CSI) of communication channels is perfectly known at the BS via a proper channel estimation mechanism. The transimit beamforming matrix is
\begin{equation}
    \mathbf{W} = [ \underbrace{\mathbf{w}_1, \dots, \mathbf{w}_K}_{{\text{for communication and sensing}}}, \underbrace{\mathbf{w}_{K+1}}_{{\text{dedicated for sensing}}}] \in \mathbb{C}^{N_t \times (K+1)}.
\end{equation}
Then the received signal at the $k$-th user is given by
\begin{equation}
\begin{aligned}
    y_k &= \mathbf h^H_k \mathbf {FWs}+n_k \\
    &=\underbrace{\mathbf h^H_k \mathbf F \mathbf w_k s_k}_{\text{Desired signal}} + \underbrace{\mathbf h^H_k \mathbf F \mathbf w_{k'} s_{k'}}_{\text{Interference}} +n_k,
\end{aligned}
\end{equation}
where $\mathbf{w}_k \in \mathbb{C}^{N_t \times \text{1}}$ denotes the digital transmit beamforming vector for the user $k \in \mathcal{K} $, $\mathbf F\in \mathbb C^{N_t \times N_{RF}}$ is the analog precoding matrix, and $n_k \sim \mathcal{C} \mathcal{N} (0,\sigma_k^2)$ is zero-mean additive white Gaussian noise (AWGN) with noise power $\sigma_k^2$. With this sub-connected structure, the analog precoder matrix $\mathbf F$ is given by
\begin{equation}
\begin{aligned}
    \mathbf F=\mathrm {Bdiag}(\mathbf f_1,\mathbf f_2,\ldots,\mathbf f_{N_{RF}})\in \mathbb C^{N_t \times N_{RF}},
\end{aligned}
\end{equation}
where $\mathbf f_i \in \mathbb C^{M \times 1},i=1,2,\dots,N_{RF}$, is an $M=N_t/N_{RF}$ dimensional vector with each element being constant modulus value, i.e, $|\mathbf f_i(j)|=1,j=1,\dots,M$.
Hence, the received data rate of user $k$ can be obtained as
\begin{equation}\label{Rk}
\begin{aligned}
    R_k=\mathrm {log}_2(1+ \text{SINR}_k),
\end{aligned}
\end{equation}
where 
\begin{equation}
\begin{aligned}
    \text{SINR}_k=\frac{|\mathbf h^H_k \mathbf {Fw}_k|}{\sum^{K+1}_{j=1,j\neq k}|\mathbf h^H_k \mathbf {Fw}_j|^2+\sigma_k^2}.
\end{aligned}
\end{equation}

As for sensing, the BS applies the receive beamformer, $\mathbf u \in \mathbb C^{N_r \times 1}$ to capture the reflected signal of target $s$. The received signal at the BS can be expressed as
\begin{equation}
    y_s= \underbrace{\zeta_s \mathbf u^\mathrm H \mathbf G_s \mathbf{F W s}}_{\text {Target reflection}}+\underbrace{\sum^C_{c=1} \zeta_c \mathbf u^\mathrm H \mathbf G_c \mathbf{F W s}}_{\text {Clutters reflection}}+n_s,
\end{equation}
where $n_s \sim \mathcal{C} \mathcal{N} (0,\sigma_s^2)$ denotes the AWGN for radar link. $\zeta_s$ and $\zeta_c $ are complex coefficients including the RCS of the target/clutters, respectively. Additionally, $\mathbf G_s$ and $\mathbf G_c$ denote the array response vectors between the BS and the target/clutters, respectively.

We assume point-like target and clutter, since the sizes of the target and the clutter are sufficiently small compared to the distances of the reflecting paths. Thus, the radar signal-to-clutter-plus-noise-ratio (SCNR) at the RP can be expressed as
\begin{equation}\label{SCNR}
    \text{SCNR} = \frac{||\zeta_s \mathbf u^\mathrm H \mathbf G_s \mathbf{F W}||^2}{||\sum^C_{c=1} \zeta_c \mathbf u^\mathrm H \mathbf G_c \mathbf{F W}||^2+\sigma^2_s}.
\end{equation}

For point target detection in MIMO radar systems, to quantify how much information can be extracted with a given sensing signal, we can use the sensing mutual information (MI) per unit time \cite{Ouyang2023}\cite{Peng2024}. It is defined as the MI between the received signal at the BS and the target response, conditioned on the transmitted signal. Specifically, the BS has knowledge of both the receive beamformer $\mathbf u$ and the transmitted ISAC signal $\mathbf x$, the conditional sensing MI between the received echo and the target response (or the sensing rate) as \cite{Lyu2025}
\begin{equation}\label{Rs}
    R_s = I(y_s;\zeta_s\mathbf G_s|\mathbf u,\mathbf x) =\text {log}_2(1+ \text {SCNR})
\end{equation}
where $\mathbf x = \mathbf {FWs}$ denotes the transmitted signal.

\subsection{Problem Formulation}
To balance the communication and sensing performance, we aim to maximize the weighted sum of communication rate and sensing MI in \eqref{Rk} and \eqref{Rs}. The optimization problem is formulated as 
\begin{subequations}\label{P_1}
\begin{align}
    \max_{{\mathbf{u}},{\mathbf{F}},{\mathbf{W}},{\Lambda  }}&~{\mathcal G({\mathbf{u}},{\mathbf{F}},{\mathbf{W}},{\mathbf \Lambda})=\varpi_c \sum^K_{k=1}R_k + \varpi_s R_s} \label{G}\\
    {\text{s.t.}}&~||\mathbf u||^2_2 = 1,\label{u constraint} \\ 
    &~||\mathbf{FW}||^2_F\leq P_{max},\label{P constraint}\\ 
    &~\mathbf F\in \mathcal A_F,\label{F constraint}
    \end{align}
\end{subequations}
where the weighting factors $\varpi _s$ and $\varpi _c$ control the priority of communication and sensing, and satisfy $\varpi _s + \varpi _c =1$. $\mathbf \Lambda =[\alpha^\mathrm t, \beta^\mathrm t, \gamma^\mathrm t,\alpha^\mathrm r, \beta^\mathrm r, \gamma^\mathrm r]$ represents the sets of rotation angles for the transmit and receive panels. 
Constraint \eqref{u constraint} normalizes the receive beamformer. Constraint \eqref{P constraint} denotes the transmit power should not exceed the maximum power limit. Constraint (\ref{F constraint}) corresponds to the sub-connected structure, i.e., the analog $\mathbf F$ belongs to a set of block matrices $\mathcal A_p$, where each block is an $M$ dimension vector with unit modulus elements.
Due to the non-concave/non-convex objective function/constraints, problem \eqref{P_1} is difficult to solve. 
Specifically, in contrast to the existing FPA-based hybrid beamforming designs, the position changes caused by the rotatability of transmitting and receiving antennas complicate the variable coupling between the numerators and denominators in the achievable rates of users and targets.
Moreover, since the resource allocation needs to meet the different demands of the two functions, the trade-off between communication and sensing needs to be investigated.

It is worth mentioning that, for simplicity, this paper focuses on a single target in this paper. Since the proposed algorithm can treat the echoes from other targets as radar interference in the denominator of (\ref{SCNR}), it can be similarly adapted to address the multi-target problem.
\section{Proposed Algorithm}
In this section, we propose an efficient algorithm for solving problem \eqref{P_1} by optimizing the transmit digital beamformer, the transmit analog beamformer, the receive beamformer, and the rotation angles of the transmitting and receiving panels. 
The problem is non-convex due to its objective function and coupling variables. To tackle this problem, an alternating optimization (AO) algorithm be proposed that updates each variable with fixed values of the other variables obtained from the last iteration. 

Given the mutual coupling between the digital beamformer and the analog beamformer in the objective and constraint \eqref{P constraint}, it is generally difficult to optimize these two parameters simultaneously. Moreover, the non-convex SINR, SCNR, and unit-modulus constraints make the problem even more challenging. A typical method is to first design the fully-digital beamformer and then approach the fully-digital one via hybrid beamformer design \cite{Yu2016}\cite{WangX20221}. However, the obtained solution based on such a two-stage method cannot guarantee the original constraints to be satisfied.
Therefore, we propose to directly optimize the digital beamformer and analog beamformer based on the principle of AO. 
\begin{figure*}[b] 
    \centering
    \hrulefill
    \begin{align}
        \tilde{\mathcal{G}}(\mathbf{F}, \mathbf{W}, \boldsymbol{\mu}, \boldsymbol{\xi}^c, \boldsymbol{\xi}^s ,\mathbf{\Lambda}) 
        &= \varpi_c \sum_{k=1}^K \text{log}(1 + \mu_k) + \varpi_s \text{log}(1 + \mu_{K+1}) - \varpi_c \sum_{k=1}^K \mu_k - \varpi_s \mu_{K+1} \notag \\
        & \quad + \varpi_c \sum_{k=1}^K \Bigl( 2\sqrt{1 + \mu_k} \mathrm{Re}\bigl\{ \xi_k^c \mathbf{h}_k^H(\mathbf{\Lambda}) \mathbf F\mathbf{w}_k \bigr\} 
        - \lvert \xi_k^c \rvert^2 \Bigl( \sum_{j=1}^{K+1} \lvert \mathbf{h}_k^H(\mathbf{\Lambda}) \mathbf F\mathbf{w}_j \rvert^2 + \sigma_k^2 \Bigr) \Bigr) \notag \\
        & \quad + \varpi_s \Bigl( 2\sqrt{1 + \mu_{K+1}} \mathrm{Re}\bigl\{ \zeta _s \mathbf u^H \mathbf{G}_s(\mathbf{\Lambda}) \mathbf{F} \mathbf W \boldsymbol{\xi}^s \bigr\} \notag \\  % 拆分成两行，此行隐藏编号
        & \qquad - \lVert \boldsymbol{\xi}^s \rVert^2 \Bigl( \sum_{c=1}^C \lVert \zeta_c \mathbf u^H \mathbf{G}_c(\mathbf{\Lambda}) \mathbf{F} \mathbf{W}\rVert^2 
        + \lVert \zeta_s \mathbf u^H \mathbf{G}_s(\mathbf{\Lambda}) \mathbf{F} \mathbf{W}\rVert^2 + \sigma_s^2 \Bigr) \Bigr)  % 最后一行保留编号
        \label{FP}
    \end{align}
\end{figure*}
For the variables $\mathbf F$, $\mathbf W$, $\mathbf \Lambda$, the fractional programming (FP) approach is employed\cite{ShenK2018}. We introduce auxiliary variables $\boldsymbol {\mu} =[ \mu_1,\dots,\mu_{K+1}]$, $\boldsymbol {\xi}^c =[ \xi^c_1,\dots, \xi^c_K]^T$ and $\boldsymbol {\xi}^s =[ \xi^s_1,\dots, \xi^s_{K+1}]^T$  to transform \eqref{G} into an equivalent convex form \eqref{FP}. Based on the AO framework, the problem \eqref{P_1} can be decomposed as five sub-problems. The details of the proposed algorithms are presented below.

\subsection{Receive Beamforming Optimization}
Given $\{\mathbf F,\mathbf W, \mathbf \Lambda \}$, the optimization of $\mathbf u $ only affects the SCNR. Therefore, maximizing the WSR, i.e., objective value \eqref{G}, is equivalent to maximizing SCNR. 
Hence, the subproblem for $\mathbf u$ is formulated by 
\begin{subequations}\label{u_optimize}
    \begin{align}
    {\mathcal{SP}}_1:&~\max_{{\mathbf{u}}}~\mathrm {SCNR} \label{P u}\\
    {\text{s.t.}}&~\eqref{u constraint}  .\nonumber
    \end{align}
\end{subequations}

Based on the expression in \eqref{u_optimize}, the maximization of SCNR belongs to the problem of generalized Rayleigh quotient. By applying \textbf{Proposition \ref{Proposition:1}}, the optimal  closed-form solution can be obtained.
\begin{proposition}\label{Proposition:1}
The optimal solutions of problem \eqref{u_optimize} is given by   
\begin{equation}
    \mathbf{u}^* = \frac{\left( \sum^{C}_{c=1} \tilde{\mathbf{f}}_c \tilde{\mathbf{f}}_c^H + \sigma_{\text{s}}^2 \mathbf{I}_{N_r} \right)^{-1} \mathbf{b}_s}{\left\| \left( \sum^{C}_{c=1} \tilde{\mathbf{f}}_c \tilde{\mathbf{f}}_c^H +\sigma_{\text{s}}^2 \mathbf{I}_{N_r} \right)^{-1} \mathbf{b}_s \right\|_2},
    \label{eq:u_optimal}
\end{equation}
where $\mathbf{b}_s = \mathbf{G}_s ( \sum^{K+1}_{k=1} \mathbf F \mathbf{w}_k ) \in \mathbb{C}^{N_r \times 1}$ and $\tilde{\mathbf{f}}_c = \zeta _c \mathbf G_c \mathbf{FW} \in \mathbb{C}^{Nr \times (K+1)}$.
\end{proposition}
\begin{proof}
    Rewrite the numerator part of Formula (20) as
\begin{equation}
\begin{split}
|\zeta_s|^2 \mathbf{u}^H &\mathbf{G}_s \tilde {\mathbf A} \mathbf{G}^H_s \mathbf{u} \\
&= |\zeta_s|^2 \mathbf{u}^H \mathbf{g}_s \mathbf{h}_s^H \tilde {\mathbf A} \mathbf{h}_s\mathbf{g}_s^H \mathbf{u} \\
&= |\zeta_s|^2 \mathbf{h}_s^H \tilde {\mathbf A} \mathbf{h}_s \mathbf{u}^H \mathbf{g}_s \mathbf{g}_s^H \mathbf{u}.
\end{split}
\end{equation}

Since $|\zeta_s|^2 \mathbf{h}_s^H \tilde {\mathbf A} \mathbf{h}_s$ is non-negative and independent of the receive beamformer $\mathbf u$, based on the research results in Reference \cite{GH1996}, the generalized Rayleigh quotient is applied to maximize problem \eqref{u_optimize} and arrive at the optimal solution $\mathbf u^*$, which is a minimal mean-square error (MMSE) filter.
\end{proof}
    
\subsection{Transmit Digital Beamforming Optimization}
Given $\{\mathbf u,\mathbf F, \mathbf \Lambda , \boldsymbol\mu, \boldsymbol \xi^s, \boldsymbol \xi^c\}$, we optimize the digital beamformer $\mathbf W$. To decouple the product of $\mathbf F$ and $\mathbf W$ in the original power constraint \eqref{P constraint}, we reformulate the power constraint in accordance with the sub-connected structure as follows:
\begin{equation}
    ||\mathbf{FW}||^2_F=\frac{N_t}{N_{RF}}||\mathbf W||^2_F=M||\mathbf W||^2_F.    
\end{equation}
Thus, the subproblem of transmit digital beamformer design is formulated as
\begin{subequations}\label{P_W}
\begin{align}
    \max_{{\mathbf{W}}}&~{\tilde{\mathcal G}({\mathbf{u}},{\mathbf{F}},{\mathbf \Lambda},\boldsymbol\mu, \boldsymbol \xi^s, \boldsymbol \xi^c)}\\
    {\text{s.t.}}&~||\mathbf{W}||^2_F\leq \frac{P_{max}}{M}.
\end{align}
\end{subequations}

Since $\tilde{\mathcal G}$ is a convex function with respect to $\mathbf W$, the Lagrange dual method can be employed to obtain the closed-form expression of $\mathbf W$. The Lagrangian function is defined as 
\begin{equation}
\begin{split}
\mathcal{L}({\mathbf W,\lambda})&=-\tilde{\mathcal{G}}(\mathbf W |\mathbf{u},\mathbf{F},\mathbf \Lambda,\boldsymbol\mu, \boldsymbol \xi^s, \boldsymbol \xi^c) \\
&+ \lambda(||\mathbf W||^2_F-\frac{P_{max}}{M}),
\end{split}
\end{equation}
where $\lambda \geq  \text{0}$ is the Lagrange multiplier corresponding to the power constraint. The KKT conditions are used to solve the dual problem as

\begin{align}
\frac{\partial \mathcal{L}(\mathbf{W},\lambda)}{\partial \mathbf{W}} &= \mathbf{0}, \label{kkt}\\
\|\mathbf{W}\|_F^2 - \frac{P_{\text{max}}}{M} &\leq 0, \label{kkt1}\\
\lambda &\geq 0, \label{kkt2}\\
\lambda\left( \|\mathbf{W}\|_F^2 - \frac{P_{\text{max}}}{M} \right) &= 0. \label{kkt3}
\end{align}

The closed-form expression of $\mathbf W$ can be obtained by solving \eqref{kkt}, which is
\begin{equation}
\mathbf{w}_k(\lambda) = \left( \left( \boldsymbol{\Lambda}_k^T + \lambda \mathbf{I}_{N_{\text{RF}}} \right)^{-1} \right)^* \boldsymbol{\varphi}_k, \quad \forall k \in \{1, \dots, K+1\},
\end{equation}
where
\begin{align}
\boldsymbol{\Lambda}_k &= \varpi_c \sum_{k=1}^K |\xi_k^c|^2 (\mathbf{h}_k^H \mathbf{F})^H (\mathbf{h}_k^H \mathbf{F}) \notag \\
&+ \varpi_s \|\boldsymbol{\xi}^s\|^2 \biggl\{ 
\lvert \zeta_s \rvert^2 ( \mathbf{u}^H \mathbf{G}_s \mathbf{F} )^H 
( \mathbf{u}^H \mathbf{G}_s \mathbf{F} ) \notag \\
&\quad + \sum_{c=1}^{C} \lvert \zeta_c \rvert^2 ( \mathbf{u}^H \mathbf{G}_c \mathbf{F} )^H 
( \mathbf{u}^H \mathbf{G}_c \mathbf{F} )
\biggr\}, \\
\boldsymbol{\varphi}_k &= \varpi_c \sqrt{1 + \mu_k} \, (\xi_k^{c} \mathbf{h}_k^H \mathbf F)^H + \varpi_s \sqrt{1 + \mu_{K+1}} \notag \\
&\quad \times (\zeta_s \xi_k^{s} \mathbf u^H \mathbf{G}_s \mathbf F)^H, \quad \text{for } k \in \{1, \dots, K\}, \\
\boldsymbol{\varphi}_{K+1} &= \varpi_s \sqrt{1 + \mu_{K+1}}(\zeta_s \xi_{K+1}^{s} \mathbf u^H \mathbf{G}_s \mathbf{F})^H.
\end{align}

We employ a bisection method to select an appropriate dual variable $\lambda$ to satisfy the complementary slackness condition \eqref{kkt3}. If the peimal feasibility satisfied $||\mathbf W^*||_F^2 \leq \frac{P_{max}}{M}$, $\lambda = \text{0}$. Otherwise, $\lambda$ needs to be decided to satisfy

\begin{equation}
    \mathbf h(\lambda)=||\mathbf W(\lambda)||_F^2-\frac{P_{max}}{M}\leq \varepsilon .
 \end{equation}

The bisection method can be adopted to find the solution of $\lambda$ \cite{PanC2020}, which is summarized in $\textbf{Algorithm}$ \ref{alg:bisection}.
\begin{algorithm}[!t]
\caption{Bisection Method for Searching Dual Variable $\lambda$}
\label{alg:bisection}
\begin{algorithmic}[1]
\State \textbf{Initialize} upper and lower bound $\lambda_{\max}$, $\lambda_{\min}$, tolerance $\varepsilon$ and iteration index $l = 0$.
\Repeat
\State Compute $\lambda^{(l)} = (\lambda_{\min} + \lambda_{\max}) / 2$.
\State Replace $\lambda^{(l)}$ in $\mathbf{W}(\lambda)$ and compute $h(\lambda^{(l)})$.
\If{$h(\lambda^{(l)}) > 0$}
    \State Set $\lambda_{\min} = \lambda^{(l)}$.
\Else
    \State Set $\lambda_{\max} = \lambda^{(l)}$.
\EndIf
\State Set iteration index $l = l + 1$.
\Until{$|h(\lambda^{(l)})| \leq \varepsilon$}
\State \textbf{Output}: optimal dual variable $\lambda^\star$.
\end{algorithmic}
\end{algorithm}

\subsection{Transmit Analog Beamforming Optimization}
With given $\{\mathbf u,\mathbf W, \mathbf \Lambda , \boldsymbol\mu, \boldsymbol \xi^s, \boldsymbol \xi^c\}$, we optimize the analog beamformer $\mathbf F$. The subproblem of analog beamformer design is formulated as
\begin{subequations}\label{P_F}
\begin{align}
    \max_{{\mathbf{F}}}&~{\tilde{\mathcal G}({\mathbf{u}},{\mathbf{W}},{\mathbf \Lambda},\boldsymbol\mu, \boldsymbol \xi^s, \boldsymbol \xi^c)}\\
    {\text{s.t.}}&~\eqref{F constraint}.\notag
\end{align}
\end{subequations}

Note that the analog beamformer $\mathbf F$ is block diagonal, i.e., most of its elements are zeros. To reduce the computational complexity, we extract the non-zero elements of $\mathbf F$ as
\begin{equation}
    \mathbf d=[\mathbf f^T_1,\dots,\mathbf f^T_{N_{RF}}]^T \in \mathbb C^{N_t \times 1}.
 \end{equation}

The singular value decomposition (SVD) of $\mathbf W$ with a rank of $R$ is 
\begin{equation}
\mathbf{W} = \sum_{r=1}^R \rho_r \mathbf{u}_r \mathbf{v}^H_r,
 \end{equation}
where $\rho_r$, $\mathbf{u}_r$ and $\mathbf{v}_r$ are the $r$-th eigenvalue ,left and right singular matrix, respectively. Define $\tilde{\mathbf{u}}_r=\sqrt{\rho_r}\mathbf{u}_r$, $\tilde {\mathbf{v}}_r=\sqrt{\rho_r}\mathbf{v}_r$ and $\tilde V_j = \tilde v_j^H \tilde v_j$. Then by applying the matrix transformation, the problem \eqref{P_F} can be recast as
\begin{subequations}\label{F_optimize}
\begin{align}
\max_{\mathbf{d}} \quad &\,\tilde{\mathcal P} =\varpi_c \sum_{k=1}^{K} \Biggl( 
2\Re \bigl\{ \tilde{\boldsymbol{\beta}}_k^H \mathbf{d} \bigr\} 
- |\xi_k^c|^2 \sum_{j=1}^{K+1} \bigl\lvert \tilde{\mathbf{h}}_{k,j}^H \mathbf{d} \bigr\rvert^2 
\Biggr) \notag \\
&+ \varpi_s \Biggl( 
2\Re \bigl\{ \tilde{\boldsymbol{\beta}}_s^H \mathbf{d}  \bigr\} 
- \|\boldsymbol{\xi}^s\|^2 \biggl( \sum_{c=1}^{C} \bigl\| \sum_{r=1}^{R} \mathbf{u}^H \tilde{\mathbf{G}}_{c,r} \mathbf{d} \tilde{\mathbf{v}}^H_r \bigr\|^2 \notag \\
&\quad + \bigl\| \sum_{r=1}^{R} \mathbf{u}^H \tilde{\mathbf{G}}_{s,r} \mathbf{d} \tilde{\mathbf{v}}^H_r \bigr\|^2 \biggr) 
\Biggr) + \text{const}(\boldsymbol{\mu}) \label{eq:F_optimize}\\
\text{s.t.} \quad &\bigl\lvert [\mathbf{d}]_n \bigr\rvert = 1, \quad \forall\, n = 1, 2, \dots, N_t, \label{d_constraint}
\end{align}
\end{subequations}
where const($\boldsymbol \mu$) is a constant term when $\boldsymbol \mu$ is fixed, $\tilde{\mathbf{h}}_{k,j}^H = \mathbf{h}_k^H ( \text{diag}(\mathbf{w}_{j}) \otimes \mathbf{I}_M )$, 
$\tilde{\boldsymbol{\beta}}_k^H = \sqrt{1+\mu_k} \, \xi_k^{c*} \, \tilde{\mathbf{h}}_{k,k}^H$,
$\tilde{\boldsymbol{\beta}}_s^H = \sqrt{1+\mu_{K+1}}\zeta_s^* \mathbf u^H \mathbf G_s (\text{diag}(\mathbf W \boldsymbol \xi_s)\otimes \mathbf I_M)$ and $\tilde{\mathbf{G}}_{c,r} = \zeta_c \mathbf G_c(\text{diag}(\tilde {\mathbf u}_r)\otimes \mathbf I_M)$.
Problem \eqref{F_optimize} is subject to the unit modulus constraint \eqref{d_constraint}. A conventional approach is to adopt Riemannian manifold optimization and solve it using the conjugate gradient descent method. However, this method suffers from low efficiency and unstable algorithm performance.
To address the problem \eqref{F_optimize}, the continuous variable $\phi$ is introduced. Thus the subproblem \eqref{eq:F_optimize} is formulated as 
\begin{equation}
\max~{\tilde {\mathcal P} -\eta||\boldsymbol\phi-\mathbf d||^2_2}, 
\end{equation}
where $\eta>0$ is the penalty parameter. Based on the penalty method, problem \eqref{F_optimize} can be solved via updating $\mathbf d$ and $\boldsymbol \phi$ iteratively.
With given $\mathbf d$, the optimal solution of $\boldsymbol \phi$ is
\begin{equation}
    \boldsymbol \phi^*= [ \tilde{\boldsymbol{\Xi}}]^{-1} \tilde {\boldsymbol \beta},
\end{equation}
where $\tilde{\boldsymbol{\Xi}} = \varpi_c \sum_{k=1}^K |\xi_k^c|^2 \sum_{j=1}^{K+1} \bigl( \tilde{\mathbf{h}}_{k,j} \tilde{\mathbf{h}}_{k,j}^H \bigr) 
+ \varpi_s \|\boldsymbol{\xi}^s\|^2 ( \sum_{c=1}^C \sum_{r=1}^R \tilde{V}_j \tilde{\mathbf{G}}_{c,r}^H \tilde{\mathbf{G}}_{c,r} 
+ \sum_{r=1}^R \tilde V_j \tilde{\mathbf{G}}_{s,r}^H \tilde{\mathbf{G}}_{s,r} ) + \eta \mathbf{I}_M$, $\tilde{\boldsymbol{\beta}}= \varpi _c\sum^K_{k=1}\tilde{\boldsymbol{\beta}}_k+\varpi _s\tilde{\boldsymbol{\beta}}_s+\eta \mathbf d$. With given $\boldsymbol \phi$, the optimization of $\mathbf d$ can be obtained from \eqref{F_optimize} as
\begin{subequations}\label{Phi_optimize}
\begin{align}
    \min_{{\mathbf{d}}}&~{||\boldsymbol \phi - \mathbf d||^2_2}\\
    {\text{s.t.}}&~\eqref{d_constraint}.\notag
\end{align}
\end{subequations}
The optimal phase of $\mathbf d$ is given by
\begin{equation}
    \text{arg}\{\mathbf d\}= \text{arg}\{\boldsymbol \phi\}.
\end{equation}

Finally, by performing a simple matrix transformation on $\mathbf d$, the optimized variable $\mathbf F$ can be obtained.

\subsection{Auxiliary Variables Optimization}

Since $\tilde {\mathcal G}(\boldsymbol\xi_c,\boldsymbol\xi_s , \boldsymbol \mu| \mathbf u, \mathbf F, \mathbf W, \boldsymbol \Lambda)$ is concave w.r.t. auxiliary variables $\boldsymbol\xi_c$ , $\boldsymbol\xi_s$ for quadratic transform  and auxiliary variable $\boldsymbol \mu$ for lagrangian dual transform. 
With given fixed other parameters, we can obtained the closed-form solutions of $\boldsymbol\xi_c$ , $\boldsymbol\xi_s$ and $\boldsymbol \mu$ by setting the patio derivatives to zeros, i.e.,$\frac{\partial \tilde{\mathcal G}(\boldsymbol\xi_c| \mathbf u, \mathbf F, \mathbf W, \boldsymbol \Lambda, \boldsymbol\xi_s , \boldsymbol \mu)}{\partial \boldsymbol\xi_c}=0$, $\frac{\partial \tilde{\mathcal G}(\boldsymbol\xi_s| \mathbf u, \mathbf F, \mathbf W, \boldsymbol \Lambda, \boldsymbol\xi_c , \boldsymbol \mu)}{\partial \boldsymbol\xi_s}=0$ and $\frac{\partial \tilde{\mathcal G}(\boldsymbol\mu| \mathbf u, \mathbf F, \mathbf W, \boldsymbol \Lambda, \boldsymbol\xi_c , \boldsymbol \xi_s)}{\partial \boldsymbol\mu}=0$, respectively.
Thus, we can derive the closed-form solutions for the auxiliary variables $\boldsymbol \xi_c$ and $\boldsymbol \xi_s$ related to the quadratic transformation as follows
\begin{equation}\label{xic}
    \xi_k^c=\frac{\sqrt{1+\mu_k}(\mathbf h_k^H\mathbf F \mathbf w_k)^H}{\sum^{K+1}_j|\mathbf h_k^H \mathbf F \mathbf w_j|^2+\sigma_k^2},\forall k={1,\dots,K},
\end{equation}
\begin{equation}\label{xis}
    \boldsymbol \xi^s=\frac{\sqrt{1+\mu_{K+1}}(\zeta_s \mathbf u^H \mathbf G_s \mathbf F \mathbf W)^H}{\sum^{C}_{c=1}||\zeta_c \mathbf u^H \mathbf G_c \mathbf F \mathbf W||^2+||\zeta_s \mathbf u^H \mathbf G_s \mathbf F \mathbf W||^2+\sigma_s^2}.
\end{equation}

Similarly, the closed-form solution for the auxiliary variable $\boldsymbol \mu$ related to the lagrangian dual transform, which is given by
\begin{equation}\label{muk}
    \mu_k=\frac{|\mathbf h_k^H\mathbf F \mathbf w_k|^2}{\sum^{K+1}_{j\neq k}|\mathbf h_k^H \mathbf F \mathbf w_j|^2+\sigma_k^2}, \forall k={1,\dots,K},
\end{equation}
\begin{equation}\label{muK+1}
\begin{aligned}
    \mu_{K+1}=\frac{||\zeta_s \mathbf u^H \mathbf G_s \mathbf F \mathbf W||^2}{\sum^{C}_{c=1}||\zeta_s \mathbf u^H \mathbf G_c \mathbf F \mathbf W||^2+\sigma_s^2}.
\end{aligned}
\end{equation}

\subsection{Antenna Rotation Optimization}
In this subsection, we optimize the antenna rotation angle $\boldsymbol \Lambda$ with any given other variables. $\mathbf \Lambda = [\alpha^\mathrm t, \beta^\mathrm t, \gamma^\mathrm t,\alpha^\mathrm r, \beta^\mathrm r, \gamma^\mathrm r]$ is the set of rotation angles of the transmitting and receiving panels. To optimize  $\mathbf \Lambda$, a feasible approach is to adopt an alternating optimization strategy, optimizing each element one by one until all elements in $\mathbf \Lambda$ are optimized. The transmit antenna rotation optimization and the receive antenna rotation optimization have similar mathematical formulations. For simplicity, this paper focuses on the optimization of the transmitting antenna rotation, and let $\mathbf \Lambda$ contain only one variable among $\alpha_t$, $\beta_t$ and $\gamma_t$, with this variable denoted as $\tilde \Lambda$. The subproblem of antenna rotation design is formulated as
\begin{equation}\label{P_R}
    \max_{{\tilde {\Lambda}}} \quad \tilde{\mathcal G}({\mathbf{u}},{\mathbf{W}},{\mathbf{F}},\boldsymbol\mu, \boldsymbol \xi^s, \boldsymbol \xi^c).
\end{equation}

Problem \eqref{P_R} is difficult to solve directly because of its non-convexity. Here, $\tilde{\mathcal G}({\mathbf{u}},{\mathbf{W}},{\mathbf{F}},\boldsymbol\mu, \boldsymbol \xi^s, \boldsymbol \xi^c)$ is replaced by $\tilde{\mathcal G}$ for simplicity. It is noted that $\tilde{\mathcal G}$ is differentiable across the entire feasible region, so the gradient ascent (GA) method be utilized to find a local stationary point efficiently.
Derivative $\frac{\partial \tilde{\mathcal G}}{\partial \tilde \Lambda}$ can be rewritten as
\begin{equation}
\begin{aligned}
&\frac{\partial \tilde{\mathcal{G}}(\tilde \Lambda)}{\partial \tilde \Lambda} =\varpi_c \sum_{k=1}^K \sqrt{1 + \mu_k} \frac{\partial \mathcal{F}_{1,k}(\tilde \Lambda)}{\partial \tilde \Lambda} \\
&\quad - \varpi_c \sum_{k=1}^K \vert \xi_k^c \vert^2 \sum_{j=1}^{K+1} \frac{\partial \mathcal{F}_{2,k,j}(\tilde \Lambda)}{\partial \tilde \Lambda} + \varpi_s \sqrt{1 + \mu_{K+1}} \frac{\partial \mathcal{F}_3(\tilde \Lambda)}{\partial \tilde \Lambda} \\
&\quad - \varpi_s \Vert \boldsymbol{\xi}^s \Vert^2 \sum_{c=1}^C \frac{\partial \mathcal{F}_{4,c}(\tilde \Lambda)}{\partial \tilde \Lambda} - \varpi_s \Vert \boldsymbol{\xi}^s \Vert^2 \frac{\partial \mathcal{F}_5(\tilde \Lambda)}{\partial \tilde \Lambda},
\end{aligned}
\label{eq:R_derive}
\end{equation}
where
\begin{equation}
\mathcal{F}_{1,k}(\tilde \Lambda) = 2\Re\!\left\{ \xi_k^c \mathbf{h}_k^H(\tilde \Lambda) \mathbf F \mathbf{w}_k \right\},
\label{eq:F1k}
\end{equation}
\begin{equation}
\mathcal{F}_{2,k,j}(\tilde \Lambda) = \vert \mathbf{h}_k^H(\tilde \Lambda) \mathbf F \mathbf{w}_j \vert^2,
\label{eq:F2kj}
\end{equation}
\begin{equation}
\mathcal{F}_3(\tilde \Lambda) = 2\Re\!\left\{ \zeta_s \mathbf u^H \mathbf{G}_s(\tilde \Lambda) \mathbf{F} \mathbf{W} \boldsymbol{\xi}^s \right\},
\label{eq:F3}
\end{equation}
\begin{equation}
\mathcal{F}_{4,c}(\tilde \Lambda) = \Vert \zeta_c \mathbf u^H\mathbf{G}_c(\tilde \Lambda) \mathbf F \mathbf{W} \Vert^2,
\label{eq:F4c}
\end{equation}
\begin{equation}
\mathcal{F}_5(\tilde \Lambda) = \Vert \zeta_s \mathbf u^H \mathbf{G}_s(\tilde \Lambda) \mathbf{F} \mathbf W\Vert^2.
\label{eq:F5}
\end{equation}

Define $\mathbf D_{k,n_t}= \mathbf P_k - \mathbf t_{n_t}$. $\bar{\mathbf R}_v$ denotes the derivative of the rotation matrix $\mathbf R$ with respect $v$, where $v \in \{\alpha, \beta, \gamma\}$. 
The partial derivatives w.r.t. $\tilde \Lambda$ are derived as \eqref{eq:hk_derive}--\eqref{eq:F5_derive}. 
\begin{figure*}[b] 
    \centering
    \hrulefill
\begin{align}
    &\frac{\partial [\mathbf{h}_k(\tilde{\Lambda})]_{n_t}}{\partial \tilde{\Lambda}} = \frac{A\left( (-\mathbf{q}^t_{n_t} \bar{\mathbf{R}}^T_v \mathbf{u}_0 + \mathbf{D}_{k,n_t} \bar{\mathbf{R}}_v \mathbf{u}_0 ) \|\mathbf{D}_{k,n_t}\|^2 + \mathbf{D}_{k,n_t}^T \mathbf{u}^\mathrm{t} \mathbf{D}_{k,n_t} \bar{\mathbf{R}}_v \mathbf{q}_{n_t} \right)}{2 \mathbf [h_k]_{n_t}||\mathbf D_{k,n_t}||^3} [\mathbf{a}(\boldsymbol{\varpi}^t)]_{n_t} \notag \\
    &\quad \quad \quad\quad\quad+ [h_k]_{n_t} \left( -j\frac{2\pi}{\lambda} \mathbf{P}^T_k \bar{\mathbf{R}}_v \mathbf{q}_{n_t} \right) [\mathbf{a}(\boldsymbol{\varpi}^t)]_{n_t},
    \label{eq:hk_derive} \\[5pt] % 增加与下一行的间距
    &\frac{\partial \mathcal{F}_{1,k}(\tilde \Lambda)}{\partial \tilde \Lambda} = 2\Re\left\{\xi_k^c (\frac{\partial \mathbf h_k}{\partial \tilde \Lambda})^H \mathbf F \mathbf w_k \right\},\label{eq:f1_derive}\\[5pt]
    &\frac{\partial \mathcal{F}_{2,k,j}(\tilde \Lambda)}{\partial \tilde \Lambda} = (\frac{\partial \mathbf h_k}{\partial \tilde \Lambda})^H \tilde {\mathbf A}_j \mathbf h_k + \mathbf h_k^H \tilde {\mathbf A}_j (\frac{\partial \mathbf h_k}{\partial \tilde \Lambda}), \text{ where }\tilde {\mathbf A}_j= \mathbf F {\mathbf w}_j{\mathbf w}^H_j \mathbf F^H,\\[5pt]
    &\frac{\partial \mathcal{F}_3(\tilde \Lambda)}{\partial \tilde \Lambda} = 2\Re\left\{ \mathbf u^H (\frac{\partial \mathbf G_s}{\partial \tilde \Lambda})\mathbf F \mathbf W \boldsymbol \xi^s \right\},\\[5pt]
    &\frac{\partial \mathcal{F}_{4,c}(\tilde{\Lambda})}{\partial \tilde{\Lambda}} = |\zeta_c|^2 \left( \mathbf{u}^H ( \frac{\partial \mathbf{G}_c}{\partial \tilde{\Lambda}}) \tilde {\mathbf A} \mathbf G_c \mathbf u +\mathbf u^H \mathbf G_c \tilde {\mathbf A} ( \frac{\partial \mathbf{G}_c}{\partial \tilde{\Lambda}})^H \mathbf u \right) , \text{ where } \frac{\partial \mathbf{G}_c}{\partial \tilde{\Lambda}}= \beta_c \mathbf g_c(\frac{\partial \mathbf h_c}{\partial \tilde \Lambda})^H,\\[5pt]
    &\frac{\partial \mathcal{F}_5(\tilde{\Lambda})}{\partial \tilde{\Lambda}} = |\zeta_s|^2 \left( \mathbf{u}^H ( \frac{\partial \mathbf{G}_s}{\partial \tilde{\Lambda}}) \tilde {\mathbf A} \mathbf G_s \mathbf u +\mathbf u^H \mathbf G_s \tilde {\mathbf A} ( \frac{\partial \mathbf{G}_s}{\partial \tilde{\Lambda}})^H \mathbf u \right) , \text{ where }\tilde {\mathbf A}= \mathbf F {\mathbf W}{\mathbf W}^H \mathbf F^H,\label{eq:F5_derive}
\end{align}
\end{figure*}

However, the performance of the GA is highly dependent on the step size, i.e., an excessively large step size may cause oscillations in the iteration process or even failure to converge; an excessively small step size, on the other hand, may lead the algorithm to fall into a local optimum prematurely, significantly reducing the optimization efficiency.
Meanwhile, the optimization results of this method are equally sensitive to the choice of initial points, i.e., if the initial point falls within the attraction domain of a local optimal solution, the algorithm will directly converge to that local optimum and struggle to escape to a better solution.

To address this issue, we first determine a relatively optimal initial point for GA through a preliminary search, and then employ an adaptive step size strategy for iterative searching until a feasible point is found or the convergence condition is met.
Based on the above discussions, an improved GA approach is proposed to solve the problem \eqref{P_R}.

Initially, we set a discrete point set $\mathcal K$ with $K_r$ uniformly distributed points in the feasible region $\mathcal N$. Therefore, an optimal initial $\tilde \Lambda$ can be find that maximizes $\tilde {\mathcal G}$ as 
\begin{equation}
\tilde \Lambda^{(0)} = \arg \max_{k_r \in \mathcal{K}} \tilde{\mathcal{G}}, \quad k_r \in \{1, \cdots, K_r\}. \label{step1}
\end{equation}

To ensure that in each iteration of the algorithm, the parameter is updated to a feasible point and remains within the feasible region. The update rule for \(\tilde{\Lambda}^{(t+1)}\) is defined as:
\begin{equation}\label{step2}
\tilde{\Lambda}^{(t+1)} = 
\begin{cases} 
\hat{\Lambda} \mod 2\pi, & \text{if } \tilde {\mathcal{G}}(\hat{\Lambda}) \geq \tilde {\mathcal{G}}(\tilde{\Lambda}^{(t)}), \\
\tilde{\Lambda}^{(t)}, & \text{otherwise},
\end{cases}
% \label{eq:R_update}
\end{equation}
where $\hat{\Lambda}=\tilde{\Lambda}^{(t)}+\kappa \nabla \tilde{\mathcal G}(\tilde{\Lambda}^{(t)})$ with the step size $\kappa $ for GA in each iteration, and $(t)$ indicates the value obtained from the last iteration in the inner loop for rotation angle optimization. Since $\tilde \Lambda$ represents a rotation angle with a period of 2$\pi$, the operation $\hat \Lambda \mod 2\pi$ ensures it stays within $[0,2\pi)$. 
Given that the GA method is particularly sensitive to step-size selection, although using a fixed step size is relatively simple to implement, it can easily prevent the optimization results from converging. Thus, $\kappa $ is initialized with a large positive number. 
Whenever $\tilde {\mathcal{G}}(\hat{\Lambda}) \geq  \tilde {\mathcal{G}}(\tilde{\Lambda}^{(t)})$ , $\kappa $ is iteratively updated as $\kappa \leftarrow \frac{\kappa}{2}$ until $\kappa $ is reduced to $\kappa_{min}$. This step size reduction ensures sufficient refinement of $\kappa$ while preventing excessive oscillations in the iterations.
It can be observed that the sequence $\{\mathcal {\tilde G}(\tilde \Lambda^{(t)})\}$ keeps non-decreasing based on updating guidelines \eqref{step2}, leading to the convergence of the sequence.

In this subsection, we focus on the case where $\mathbf \Lambda$ contains only one parameter. However, our method can be readily extended to scenarios involving multiple parameters. Specifically, it suffices to gradually optimize each variable in $\mathbf \Lambda$ using an alternating iteration approach until the optimal $\mathbf \Lambda^*$ is obtained.
\subsection{Overall Algorithm}
Based on the algorithms presented above, the detailed overall alternating optimization algorithm for solving the original problem \eqref{P_1} in $\textbf{Algorithm}$ \ref{alg:final_algorithm}.
\begin{algorithm}[!t]
    \caption{Proposed Alternating Optimization Algorithm for Flexible Beamforming Design}
    \label{alg:final_algorithm}
\begin{algorithmic}[1]
    \State \textbf{Initialize:} $\mathbf{u}^{(0)},\mathbf{F}^{(0)}, \mathbf{W}^{(0)}, \boldsymbol{\xi}_c^{(0)}, \boldsymbol{\xi}_s^{(0)}, \mathbf{\mu}^{(0)}, \mathbf \Lambda^{(0)}, \kappa^{(0)}$, the iteration index $t$ = 0 and the maximal iteration number $I_{AO}$.
    \Repeat
        \State Obtain $\mathbf{W}^{(t)}$ via \eqref{P_W}
        \State Obtain $\mathbf{F}^{(t)}$ via \eqref{F_optimize} and \eqref{Phi_optimize}
        \State Update initial points $\mathbf{\tilde \Lambda}^{(0)}$ via \eqref{step1}
        \State Calculate the gradient $\nabla_{\tilde \Lambda} \tilde {\mathcal{G}}(\tilde \Lambda^{(t)})$ via \eqref{eq:R_derive}
        \State Initialize the step size $\kappa = \kappa^{(0)}$
        \For {each $\tilde \Lambda$ in $\mathbf \Lambda$}
        \Repeat
            \State Compute $\hat{\Lambda}=\tilde{\Lambda}^{(t)}+\kappa \nabla \tilde{\mathcal G}(\tilde \Lambda^{(t)})$
            \State Shrink the step size $\kappa \gets \frac{\kappa}{2}$
            \State Update $\tilde \Lambda^{(t+1)}$ according to \eqref{step2}
        \Until{Converges}
        \EndFor
        \State Obtain $\mathbf{u}^{(t)}$ via \eqref{eq:u_optimal}
        \State Update $\boldsymbol{\mu}^{(t)}$ via \eqref{muk} and \eqref{muK+1}
        \State Update $\boldsymbol{\xi}^{c(t)}$ and $\boldsymbol{\xi}^{s(t)}$ via \eqref{xic} and \eqref{xis}, respectively
        \State Update $t = t + 1$
    \Until The value of the objective function \eqref{P_1} converges or the maximum iteration number $I_{AO}$ is reached.\\
\textbf{Output:} $\mathbf{u}^*, \mathbf{W}^*, \mathbf{F}^*, \mathbf{\Lambda}^*$
\end{algorithmic}
\end{algorithm}

The total complexity of $\textbf{Algorithm}$ \ref{alg:final_algorithm} stems from the solving the variables $\{\mathbf u, \mathbf W, \mathbf F,\mathbf \Lambda, \mathbf \mu, \boldsymbol \xi^c,\boldsymbol \xi_s\}$. 
The computational complexity for calculating receive beamformer $\mathbf u$ is $\mathcal O(N_r^3)$ due to the matrix inversion in \eqref{eq:u_optimal}. Similarly, the computational complexities for optimizing transmit digital beamformer $\mathbf W$ and transmit analog beamformer $\mathbf F$ are both $\mathcal O(N_t^3)$. 
To update the auxiliary variables $\mathbf \mu, \boldsymbol \xi^c$ and $\boldsymbol \xi_s$, the complexities are $\mathcal O(KN_t)$, $\mathcal O(KN_t +K^2)$ and $\mathcal O(N_tK^2)$, respectively. 
The complexity for updating $\mathbf \Lambda$ mainly stems from the derivative calculations \eqref{eq:R_derive} for the $L$ elements in $\mathbf \Lambda$. Thus, the computational complexity of $\mathbf \Lambda$ is $\mathbf O(I_{GA}LN_t^2K^2)$, where $I_{GA}$ is the number of iterations for gradient ascent. Considering the dominant computational steps, the total computational complexity of $\textbf{Algorithm}$ \ref{alg:final_algorithm} is $\mathcal O(I_{AO}I_{GA}LN_t^2K^2)$, where $I_{AO}$ denotes the number of iterations of the AO algorithm.

\section{Numerical results}
In this section, we elaborate on the simulation setup and discussion of the numerical results to evaluate the performance of the proposed algorithm.
\subsection{Simulation Setup}
In the simulation, the BS is the origin of the global coordinate system. The center coordinates of the transmitting uniform planar array (UPA) and the receiving UPA are set to $[0,0,-d_0]^T$ and $[0,0,d_0]^T$, respectively. The BS is equipped with $N_t$ transmit antennas, $N_r$ receive antennas, and $N_{RF}$ RF chains. Each RF chain is connected to $M$ antennas via analog phase shifters. The antenna's positions are set according to the transmit and receive UPAs with the largest achievable apertures $D \times D$. For each UPA, the panel size is $D \times D$ with $D=10\lambda$. The BS operates at a carrier frequency of 30GHz. It serves $K$ = 4 users and senses $s=1$ target, with the interference of $C=$3 clutters. The users and clutters are randomly distributed in front of the BS, i.e., $\phi_d \sim \mathcal U(-\frac{\pi}{2},\frac{\pi}{2})$, and $\theta_d \sim \mathcal U(0,\pi)$ for d $\in \{1,\dots,K+C\}$. The sensing target is assumed to be located at $(\frac{\pi}{3},\frac{\pi}{3})$.
The reflection coefficients of both the target and the clutters, $\zeta _k$ and $\zeta _s$ , are set to 2 dB. In addition, the power budget at the BS is set as $P_{max}$ = 0 dBm. The average noise power is $\sigma^2_k = \sigma^2_s =$ -80 dBm. The penalty factor $\eta=$ 10. The preliminary step size $\kappa = 100$ and the minimum step size $\kappa_{min}$ = 0.01. It is worth noting that, unless otherwise specified, the optimized rotation variable $\mathbf \Lambda$ includes all rotation parameters of the TP and RP, i.e., $\mathbf \Lambda = [\alpha^\mathrm t, \beta^\mathrm t, \gamma^\mathrm t,\alpha^\mathrm r, \beta^\mathrm r, \gamma^\mathrm r]$. We perform 200 Monte Carlo simulations for the following results.
The simulation parameters for the proposed system are listed in Table. \ref{system_params}
\begin{table}[!t]
\centering
\renewcommand{\arraystretch}{1.3}
\caption{System Parameters}
\label{system_params}
\begin{tabular}{|c|c|}
\hline
\textbf{Parameters} & \textbf{Value} \\ \hline
Number of antennas at the BS & $N_t = N_r= 16$ \\ \hline
Number of users & $K = 4$ \\ \hline
Number of clutters & $C = 3$ \\ \hline
Wavelength & $\lambda = 0.01\ \text{m}$ \\ \hline
Size of antenna array & $D = 10\lambda$ \\ \hline
Number of RF chains & $N_{RF} = 8$ \\ \hline
Penalty factor & $\eta = 10$ \\ \hline
Preliminary step size & $\kappa =100$ \\ \hline
Minimum step size  & $\kappa_{min} =1$ \\ \hline
Power budget & $P_{max} = 0\ \text{dBm}$ \\ \hline
Position variable of the UPA center& $d_0 = 2 \text{m}$ \\ \hline
Average noise power for users \& target & $\sigma_k^2 = \sigma_s^2 = -80\ \text{dBm}$ \\ \hline
Reflection coefficient for target \& clutters& $\zeta_s =\zeta_c= 3$ \\ \hline
Azimuth of users and clutters & $\phi_{\text{d}} \in (-\frac{\pi}{2}, \frac{\pi}{2})$ \\ \hline
Elevation of users and clutters & $\theta_{\text{d}} \in (0, \pi)$ \\ \hline
Location of the target & $(\phi_s, \theta_s) = (\frac{\pi}{3}, \frac{\pi}{3})$ \\ \hline
Path loss of the users \& target & $\rho=-30\ \text{dB}$ \\ \hline

\end{tabular}
\end{table}

To verify the effectiveness of the proposed scheme with Algorithm \ref{alg:final_algorithm}, we introduce six different schemes as follows, 
\begin{enumerate}
    \item The TP and RP have fixed rotation angles, with RF chains adopting a sub-connected structure ($\textbf{FPA, SC}$).
    \item The TP and RP have fixed rotation angles, with RF chains adopting a fully-connected structure ($\textbf{FPA, FC}$).
    \item The TP and RP have fixed rotation angles, using a fully-digital beamformer ($\textbf{FPA, FD}$).
    \item The TP can execute 3D rotations and RP has fixed rotation angles, with RF chains adopting a sub-connected structure ($\textbf{Transmit RA \& receive FPA, SC}$).
    \item The TP has fixed rotation angles and RP can execute 3D rotations, with RF chains adopting a sub-connected structure ($\textbf{Transmit FPA \& receive RA, SC}$).
    \item Both the TP and RP can execute 3D rotations, with RF chains adopting a sub-connected structure ($\textbf{RA, SC}$).
\end{enumerate}

\subsection{Convergence Evaluation}
\begin{figure}[!t]
    \centering
    \includegraphics[width=\linewidth]{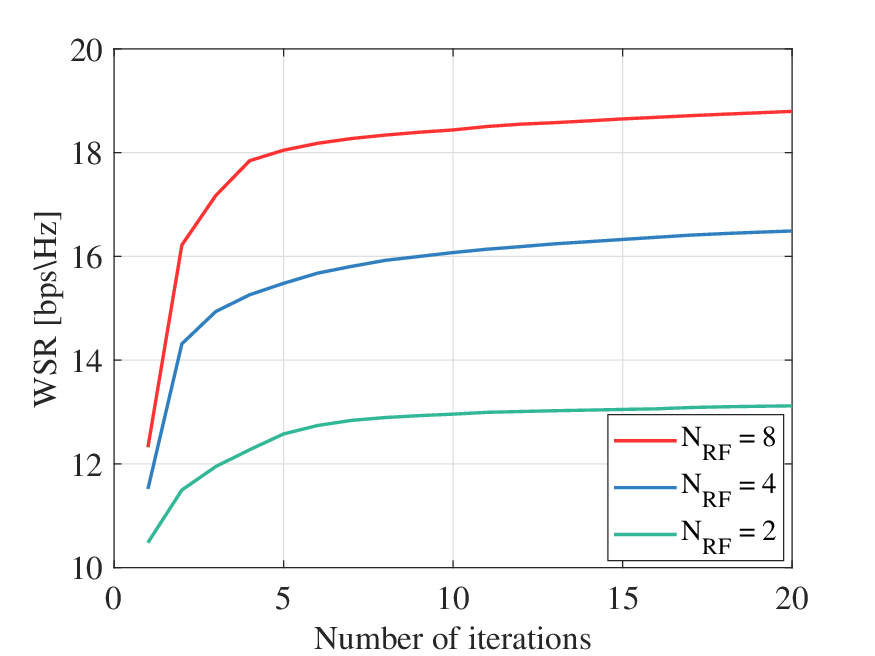}
    \captionsetup{font={small}}
    \caption{\justifying Convergence evaluation of Algorithm 2. $N_t=N_r=16$, $\varpi_c=0.5$.}
    \label{fig:convergence}
\end{figure}
First, we evaluate the convergence of the proposed AO algorithm under different numbers of RF chains in Fig. \ref{fig:convergence}. As can be seen, the algorithm converges rapidly in all scenarios, and the objective value tends to stabilize after 20 iterations.
Moreover, the WSR increases as the number of RF chains increases. This is because the performance of the analog beamformer $\mathbf F$ is constrained by the number of RF chains. Specifically, the hybrid beamformer can achieve a higher DoFs to further enhance the overall system performance when the number of RF chains increases.

\subsection{Beamfocusing of Algorithm 2}
To verify the effect of Algorithm 2 on improving sensing performance, the optimized beampattern by Algorithm 2 is compared with the optimized beampattern by FPA scheme. It intuitively demonstrates the advantages of RA in the sensing function of ISAC system proposed in this paper.
The beampattern at the angle ($\phi,\theta$) is computed as
 \begin{equation}
\text{BP}(\phi,\theta)=||\mathbf u^H \mathbf G_s(\phi,\theta)\mathbf F \mathbf W||^2.
\end{equation}
Set $\omega_s=1$ to eliminate the interference from CUs. The number of transmit antennas and receive antennas are both $N_t=N_r=16$, the number of RF chains is $N_{RF}=4$, and the target is assumed to be located at ($60^\circ,60^\circ$). Fig.3(a) shows the optimized beampattern under the FPA scheme, which indicates that the beam fails to focus well on the target position. 
In Fig.3(b), the side length of the UPA is set to $D=5\lambda$, and it is observed that under the RA-aided scheme, the beam can focus precisely on the desired target position with extremely lower interference to other positions. 
In the far-field scenario, as the side length of the UPA increases as shown in Fig.3(c), the width of the beam main lobe narrows gradually. However, due to the increase in the spacing between antennas in the array, the array periodicity is enhanced and the number of grating lobes increases, resulting in more interference. 
\begin{figure*}[!t]
    \centering  % 让整个图组在页面水平居中
    \begin{subfigure}[b]{0.3\linewidth}  % [b]底部对齐；0.3\linewidth 占行宽30%
    \centering  % 让图片在子图内水平居中
    \includegraphics[width=\linewidth]{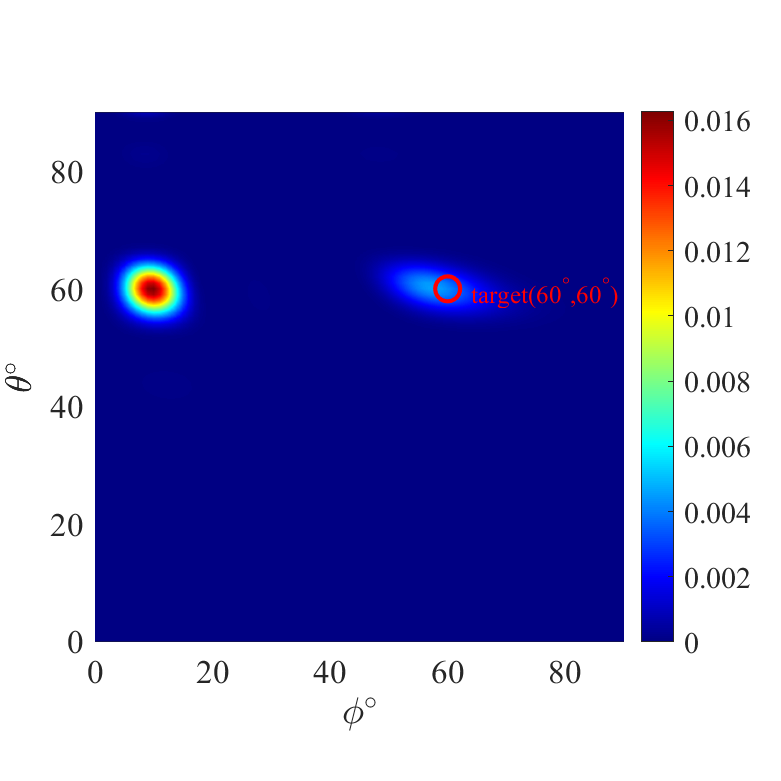}  
    \caption{FPA, $D=5\lambda$.}  % 子标题，\justifying让文字两端对齐
\end{subfigure}
\hfill  
\begin{subfigure}[b]{0.3\linewidth}
    \centering
    \includegraphics[width=\linewidth]{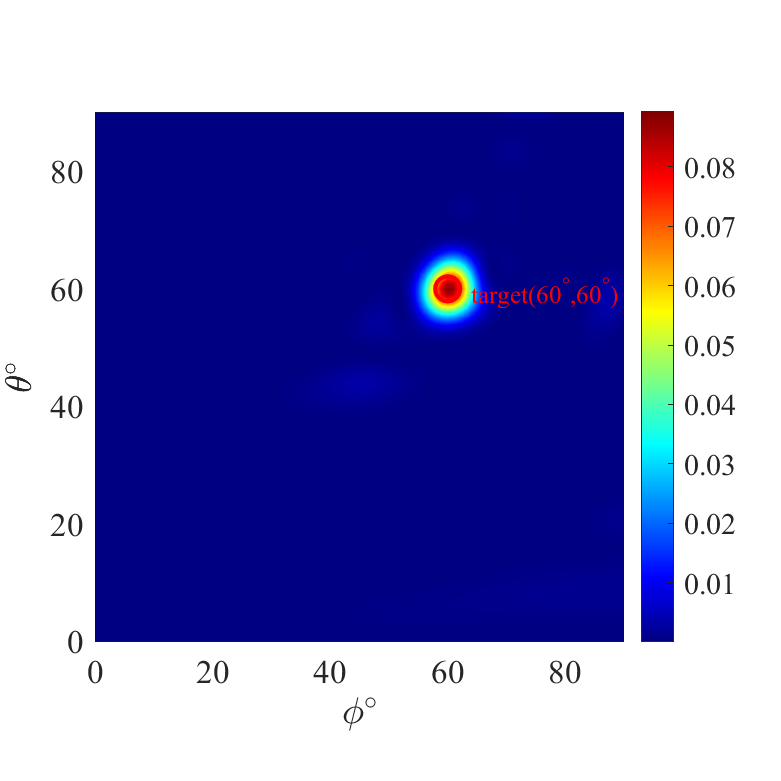}
    \caption{Proposed Algorithm, $D=5\lambda$.}
\end{subfigure}
\hfill  
\begin{subfigure}[b]{0.3\linewidth}
    \centering
    \includegraphics[width=\linewidth]{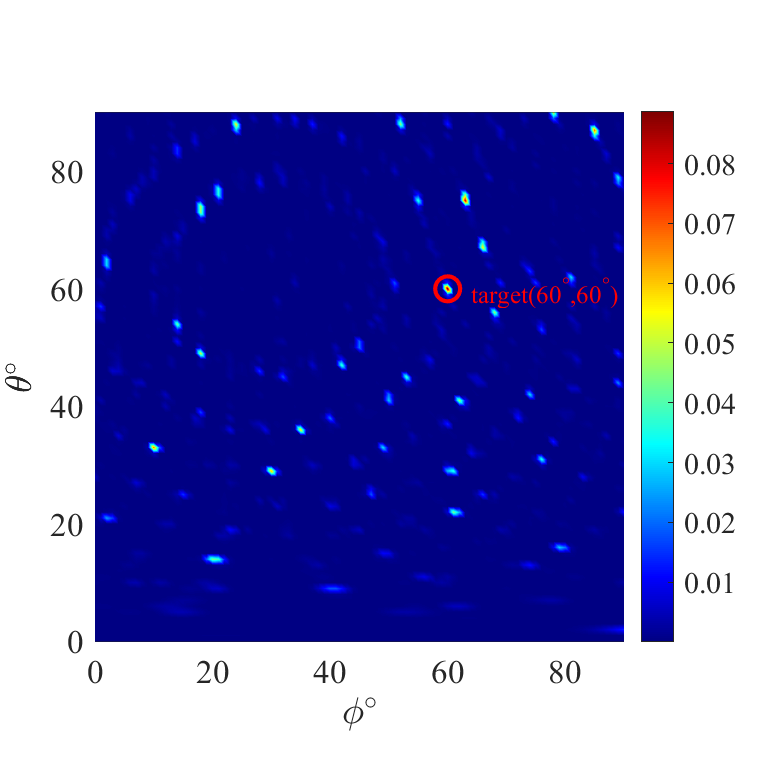}
    \caption{Proposed Algorithm, $D=50\lambda$.}
\end{subfigure}
% 整个figure的总标题（放在所有子图下方）
\captionsetup{justification=raggedright, singlelinecheck=false}
\caption{\justifying Beampattern under different configurations.}
\end{figure*}

\subsection{Simulation Results}
Fig. \ref{fig:WSR} compares the relationship between the WSR and the maximum transmit signal power $P_{max}$ among different schemes. With the increase of transmit signal power, the WSR values of both the RA and FPA schemes keep rising, which is because a higher transmit power can enhance the signal strength.
It is noteworthy that the performance of the RASC scheme is superior to that of all other schemes. Specifically, when the transmit power is 0 dBm, it achieves performance improvements of $16.96 \%$, $60.42 \%$, and $82.39 \%$ compared with the FPAFD, FPAFC and FPASC schemes, respectively. This is attributed to the RA scheme’s ability to achieve effective improvement in channel conditions and exhibit higher flexibility in beamforming design.
When only the RP is rotatable, the system can only optimize the echo channel but cannot improve the DoFs of beamforming at transmitter, resulting in limited performance improvement. When only the TP is rotatable, although it can effectively enhance user performance compared with the FPASC scheme, it is still inferior to the proposed RASC scheme. The latter realizes the optimization of the echo channel and the improvement of beamforming capability simultaneously through the rotation of TP and RP, which fully proves the effectiveness of the proposed scheme in improving the performance of the ISAC system.
\begin{figure}[!t]
    \centering
    \includegraphics[width=\linewidth]{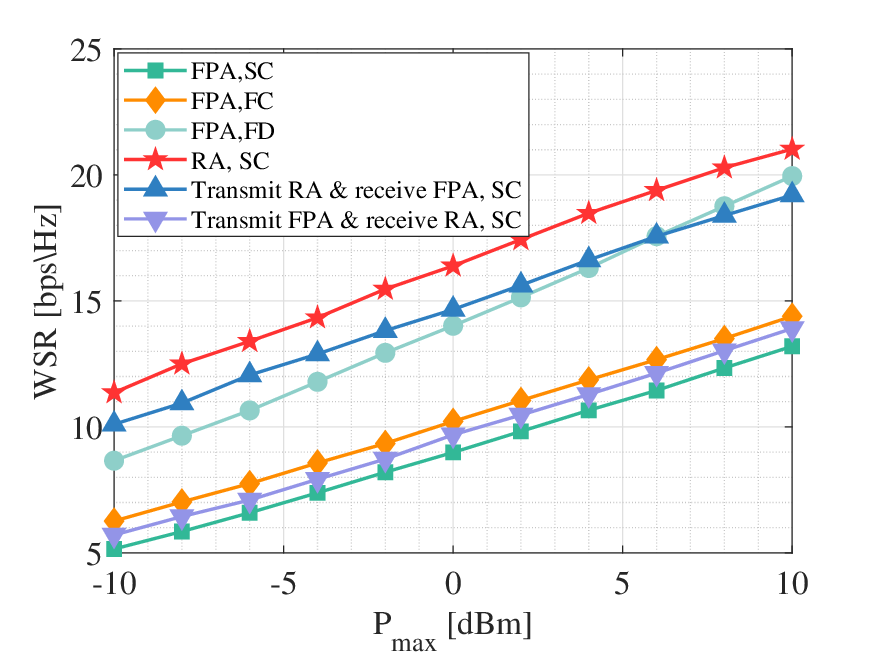}
    \captionsetup{font={small}}
    \caption{\justifying WSR versus $P_{max}$. $N_t=N_r=16$, $M=4$, $\varpi_c=0.5$.}
    \label{fig:WSR}
\end{figure}

Fig. \ref{fig:SCNR} shows the relationship between the sensing rate and the maximum transmit power $P_{max}$ under different schemes. We set $\omega_s=0.8$, which means the ISAC system is sensing-focused. As $P_{max}$ increases, the sensing rate of each scheme shows an upward trend. It can be intuitively observed from the figure that array rotation brings significant gain to sensing performance, and at the same time, the impact of different precoding structures on the sensing rate can also be clearly seen. 
In particular, although the system is sensing-focused, the sensing performance of the RARoT scheme is superior to that of the RARoR scheme, due to the rotation of the TP can effectively enhance the beamforming capability. 

Fig. \ref{fig:SINR} illustrates the sum of communication rate versus the maximum transmit power $P_{max}$ at a communication-focused ISAC with $\omega_c=0.8$. A similar trend with Fig. \ref{fig:SCNR} can be observed, and the performance improvement on communication is more significant than that of sensing. 
Notably, when the system is communication-focused, compared with FPASC scheme, the scheme that only rotates the RP will instead lead to a decrease in the system's communication capability. This is because although the rotation of the RP can effectively improve the system's sensing capability, it will cause the system to allocate more transmit power to sensing during the resource trade-off process, thereby resulting in a slight decline in the communication performance.
% \begin{figure}[!t]
%     \centering
%     \includegraphics[width=\linewidth]{VersusPSCNR.eps}
%     \captionsetup{font={small}}
%     \caption{\justifying Sensing MI versus $P_{max}$. $N_t=N_r=16$, $M=4$, $\varpi_s=1$.}
%     \label{fig:SCNR}
% \end{figure}
% \begin{figure}[!t]
%     \centering
%     \includegraphics[width=\linewidth]{VersusPSINR.eps}
%     \captionsetup{font={small}}
%     \caption{\justifying Communication rate versus $P_{max}$. $N_t=N_r=16$, $M=4$, $\varpi_c=1$.}
%     \label{fig:SINR}
% \end{figure}
% \begin{figure}[!t]
%     \centering
%     \includegraphics[width=\linewidth]{tradeoff.eps}
%     \captionsetup{font={small}}
%     \caption{\justifying Sensing MI versus $P_{max}$. $N_t=N_r=16$, $M=4$, $\varpi_s=1$.}
%     \label{fig:tradeoff}
% \end{figure}
% \begin{figure}[!t]
%     \centering
%     \includegraphics[width=\linewidth]{VersusM.eps}
%     \captionsetup{font={small}}
%     \caption{\justifying Sensing MI versus $P_{max}$. $N_t=N_r=16$, $M=4$, $\varpi_s=1$.}
%     \label{fig:M}
% \end{figure}
% \begin{figure}[!t]
%         \centering
%         \includegraphics[width=\linewidth]{VersusPSCNR.eps}
%         \captionsetup{font={small}}
%         \caption{Sensing rate versus $P_{max}$. $N_t=N_r=16$, $M=4$, $\varpi_s=0.8$.}
%         \label{fig:SCNR}
% \end{figure}
% \begin{figure}[!t]
%         \centering
%         \includegraphics[width=\linewidth]{VersusPSINR.eps}
%         \captionsetup{font={small}}
%         \caption{Sum of communication rate versus $P_{max}$. $N_t=N_r=16$, $M=4$, $\varpi_c=0.8$.}
%         \label{fig:SINR}
% \end{figure}

\begin{figure*}[!t]
    \begin{minipage}[t]{0.48\linewidth}  % 左图占48%宽度
        \centering
        \includegraphics[width=\linewidth]{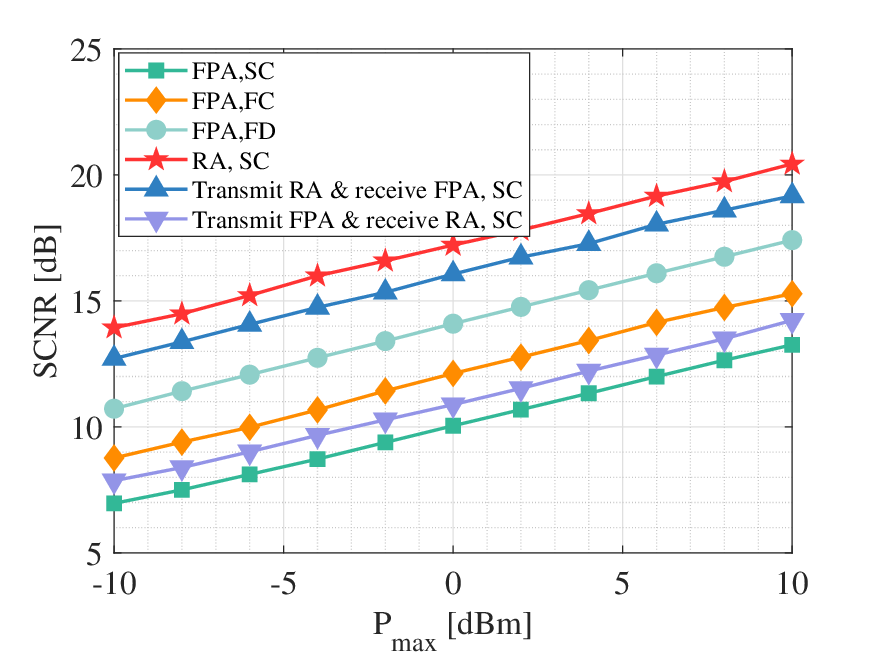}
        \captionsetup{font={small}}
        \caption{Sensing rate versus $P_{max}$. $N_t=N_r=16$, $M=4$, $\varpi_s=0.8$.}
        \label{fig:SCNR}
    \end{minipage}
    \hfill  % 两图之间留白
    \begin{minipage}[t]{0.48\linewidth}  % 右图占48%宽度
        \centering
        \includegraphics[width=\linewidth]{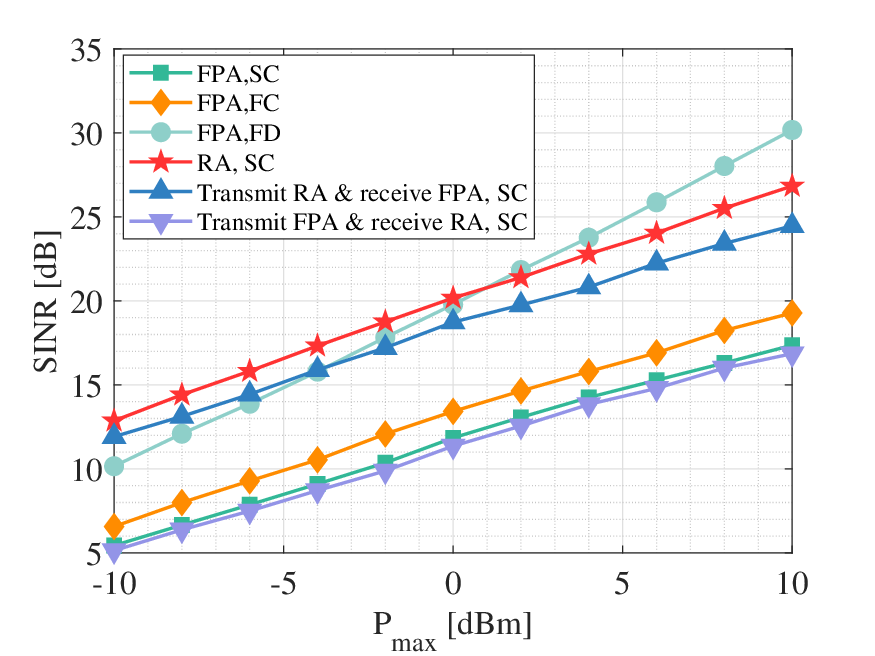}
        \captionsetup{font={small}}
        \caption{Sum of communication rate versus $P_{max}$. $N_t=N_r=16$, $M=4$, $\varpi_c=0.8$.}
        \label{fig:SINR}
    \end{minipage}
\end{figure*}

% 第二行两个图
\begin{figure*}[!t]
    \begin{minipage}[t]{0.48\linewidth}  % 左图占48%宽度
        \centering
        \includegraphics[width=\linewidth]{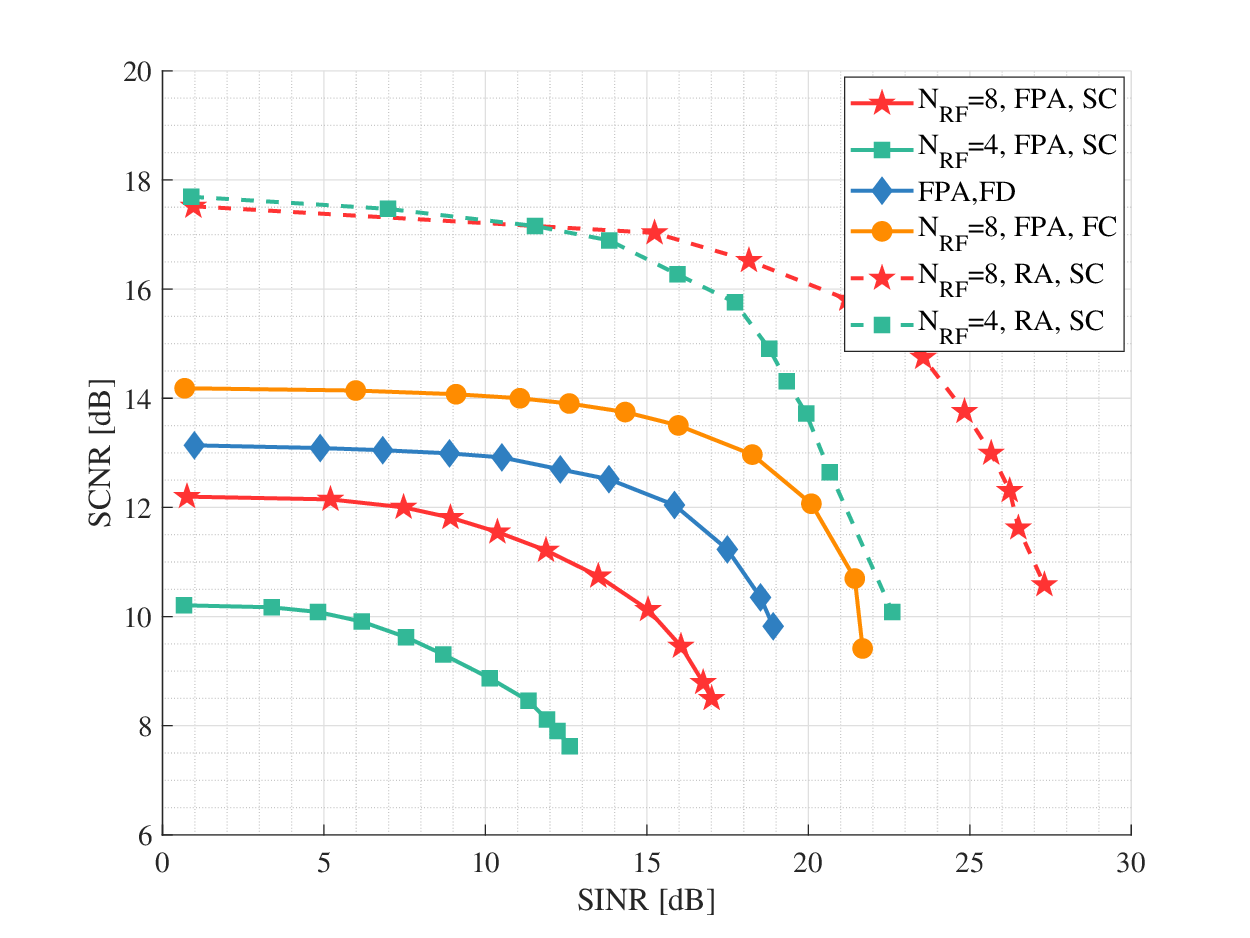}
        \captionsetup{font={small}}
        \caption{Trade-off between sensing and communication performance.}
        \label{fig:tradeoff}
    \end{minipage}
    \hfill  % 两图之间留白
    \begin{minipage}[t]{0.48\linewidth}  % 右图占48%宽度
        \centering
        \includegraphics[width=\linewidth]{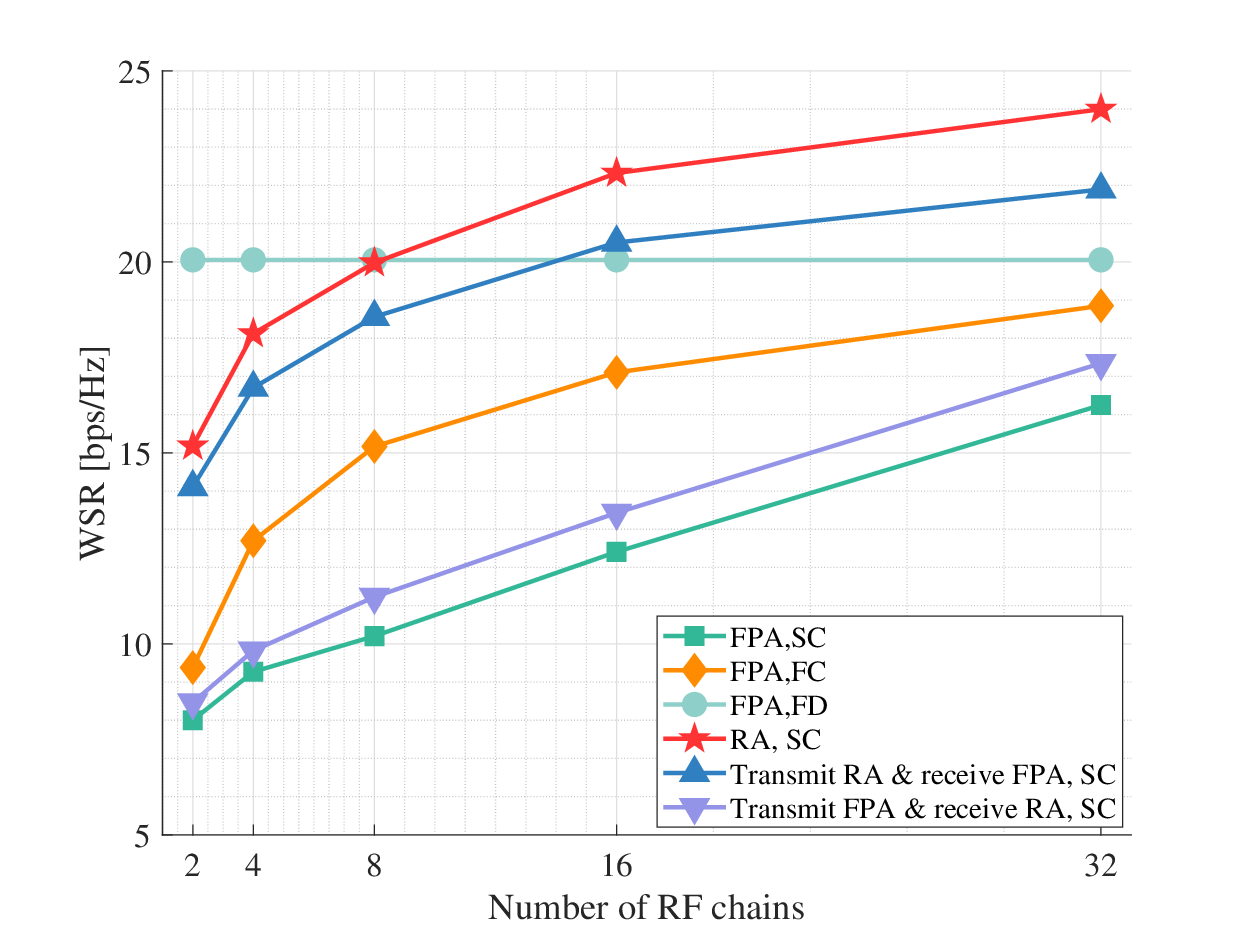}
        \captionsetup{font={small}}
        \caption{WSR versus $M$. $N_t=N_r=64$, $D=10\lambda$, $\varpi_c=0.5$.}
        \label{fig:M}
    \end{minipage}
\end{figure*}
Fig. \ref{fig:tradeoff} analyzes the trade-off between sensing and communication with $P_{max}=0$ dBm and $D=10\lambda$.
By adjusting the sensing rate weight $\omega_s$ and the communication rate weight $\omega_c$, we obtain the trade-off curves between the sensing rate and the communication rate.
To analyze the impact of RA and sub-connected structure on system performance, let $N_t=N_r=16$. Meanwhile, two sub-connection structures ($N_{RF}=4, N_{RF}=8$) and a fully-connected structure $(N_{RF}=8)$ are considered.
As shown in Fig. \ref{fig:tradeoff}, RA can significantly improve system performance compared with FPA. Under the FPA scheme, the system performance of the precoding structure is inferior to that of the fully-digital beamformer. But the performance advantage of the fully-digital beamformer is achieved at the cost of higher hardware expenses.
For the precoding structure, the more RF chains, the better the system performance. When the number of RF chains is the same, the system performance of the fully-connected structure is superior to that of the sub-connected structure. However, the fully-connected structure uses more PSs, which makes the system hardware more complex.
It is worth noting that under the RA scheme, the system performance is similar between the cases with 4 RF chains and 8 RF chains when the communication weight $\omega_c=0$. This indicates that the number of RF chains has almost no impact on the target SCNR. However, $\omega_c$ increases and the system is gradually communication-focused, the performance of the system with 8 RF chains is significantly superior to that with 4 RF chains.
Hence, appropriate weighted factors and configuration parameters need to be considered in partial engineering.

Finally, Figure. \ref{fig:M} illustrates the relationship between the WSR and the number of RF chains. 
The number of 8 $ \times $ 8 UPAs is $N_t=N_r=64$. It can be observed that as the $N_{RF}$ increases, the WSR under the hybrid beamforming design shows an upward trend. 
This can be explained from two perspectives. On one hand, an increased number of RF chains provides more DoFs for the downlink beamfocusing and thus mitigates the inter-user interference among CUs; On the other hand, with an increased number of RF chains, a more sophisticated probing signal design can be achieved to realize better positioning performance. 
When the number of RF chains exceeds 8, the system performance of the RASC scheme is superior to that of the FPA fully-digital beamforming scheme. This result indicates that the proposed RA scheme can significantly reduce the number of RF chains while still achieving performance that surpasses that of the fully-digital beamforming scheme.

\section{Conclusion}
In this paper, we study an RA-aided sub-connected hybrid beamforming design for ISAC systems. By jointly optimizing the angle of the rotatable array and the hybrid beamforming vector, we maximize the system's achievable WSR for both communications and sensing. 
An AO-based framework is proposed to solve this non-convex optimization problem. We transform the objective function using the FP method, and alternately solve five sub-problems until the AO algorithm converges. 
Firstly, the MMSE filtering method is adopted to optimize the receive beamformer. Moreover, the closed-form expressions is derived to update the digital and analog beamforming matrices. Then we derived the derivative expression of the objective function with respect to the rotation angle through research and analysis. Based on this expression, a GA method is employed to optimize the angle parameters of the rotatable array.
Simulation results verified the effectiveness of the proposed algorithm and the advantages of the considered RA-aided sub-connected structure ISAC scheme compared to FPA-based schemes. 
Remarkably, under some weight parameters, the performance of RA-aided hybrid beamforming ISAC system equipped with 64 transmit/receive antennas and 8 RF chains is outperformed to that of FPA fully-digital ISAC system equipped with 64 transmit/receive antennas.
This finding provides valuable insights for reducing hardware costs in engineering applications. In addition, apart from the sensing rate, this paper also derives and analyzes the sensing beampattern. The results show that RA can effectively optimize the beamfocusing effect and improve the sensing accuracy.

\end{document}